%% file: journal.tex
\documentclass[10pt, conference]{IEEEtran}
\IEEEoverridecommandlockouts

\usepackage{booktabs} 
\usepackage{subcaption} 
\usepackage{enumerate}
\usepackage{listings}
\usepackage{url}
\usepackage{amsmath,amsfonts,amssymb}
\usepackage{ntheorem}
\newtheorem{theorem}{Theorem}
\newtheorem{lemma}{Lemma}
\newtheorem*{proof}{Proof}

\usepackage{graphicx}
\usepackage{textcomp}
\usepackage{xcolor}
\usepackage[justification=centering]{caption}
\usepackage{verbatim}
\usepackage[linesnumbered,ruled,vlined]{algorithm2e}

\usepackage{fancyhdr}

\usepackage{array, calc}
\newcolumntype{C}[1]{>{\centering\arraybackslash}p{#1}}

\newcolumntype{R}{>{\FormatNo} c }
\def\FormatNo\ignorespaces#1\\{%
	\ignorespaces\makebox[\widthof{000}][r]{#1}\tabularnewline}

\begin{document}

\title{An Efficient and Balanced Graph Partition Algorithm for the Subgraph-Centric Programming Model on Large-scale Power-law Graphs}

\author{\IEEEauthorblockN{Shuai Zhang\IEEEauthorrefmark{2}\IEEEauthorrefmark{3}, Zite Jiang\IEEEauthorrefmark{2}\IEEEauthorrefmark{3}, Xingzhong Hou\IEEEauthorrefmark{2}\IEEEauthorrefmark{3}, Zhen Guan\IEEEauthorrefmark{2}, Mengting Yuan\IEEEauthorrefmark{4} and Haihang You\thanks{\IEEEauthorrefmark{1} Haihang You and Mengting Yuan are corresponding authors.} \IEEEauthorrefmark{2}}
	\IEEEauthorblockA{\IEEEauthorrefmark{2}SKL Computer Architecture, Institute of Computing Technology, Chinese~Academy~of~Sciences,~Beijing,~China \\
		Email: \{zhangshuai-ams, jiangzite19s, houxingzhong19z, guanzhen, youhaihang\}@ict.ac.cn} 
	\IEEEauthorblockA{\IEEEauthorrefmark{3}School of Computer Science and Technology, University of Chinese Academy of Sciences} 
	\IEEEauthorblockA{\IEEEauthorrefmark{4}School of Computer Science, Wuhan University\\
		Email: ymt@whu.edu.cn}
}
\maketitle

\begin{abstract}
Nowadays, the parallel processing of power-law graphs is one of the biggest challenges in the field of graph computation.
The subgraph-centric programming model is a promising approach and has been applied in many state-of-the-art distributed graph computing frameworks. 
The graph partition algorithm plays an important role in the overall performance of subgraph-centric frameworks. 
However, traditional graph partition algorithms have significant difficulties in processing large-scale power-law graphs. 
The major problem is the communication bottleneck found in many subgraph-centric frameworks. 
Detailed analysis indicates that the communication bottleneck is caused by the huge communication volume or the extreme message imbalance among partitioned subgraphs. 
The traditional partition algorithms do not consider both factors at the same time, especially on power-law graphs.

In this paper, we propose a novel efficient and balanced vertex-cut graph partition algorithm (EBV) which grants appropriate weights to the overall communication cost and communication balance. 
We observe that the number of replicated vertices and the balance of edge and vertex assignment have a great influence on communication patterns of distributed subgraph-centric frameworks,
which further affect the overall performance. 
Based on this insight, We design an evaluation function that quantifies the proportion of replicated vertices and the balance of edges and vertices assignments as important parameters. 
Besides, we sort the order of edge processing by the sum of end-vertices' degrees from small to large.
Experiments show that EBV reduces replication factor and communication by at least 21.8\% and 23.7\% respectively than other self-based partition algorithms. When deployed in the subgraph-centric framework, it reduces the running time on power-law graphs by an average of 16.8\% compared with the state-of-the-art partition algorithm. Our results indicate that EBV has a great potential in improving the performance of subgraph-centric frameworks for the parallel large-scale power-law graph processing.
\end{abstract}


%

\input{intro.tex}
\input{motivation}
\input{background.tex}
\input{droneapi.tex}
\input{experiments}
\input{relatedwork}
\input{conclusion}

\section*{Acknowledgment}
This work was partially supported by Natural Science Foundation
of China Grants 41930110, 61872272 and 61640221.

\bibliographystyle{abbrv}
\bibliography{journal}

\end{document}

%% file: intro.tex
\section{Introduction}
In recent decades, the need for processing large-scale graphs is increasing in both research and industry communities. 
Moreover, the parallel processing of graphs is one of the biggest challenges in the field of graph computation.
Graph computation has become vital to a wide range of data analysis tasks such as link prediction and graph pattern matching. 
Mining ``insights" from large-scale graphs is a challenging and engaging research area. 
However, the parallel graph computation poses great difficulties for graphs with billions of vertices and edges on modern high-performance computing clusters with thousands of computing nodes. The basic approaches for large-scale parallel graph computation can be divided into two programming models: the \emph{vertex-centric} model~\cite{malewicz2010pregel} and the \emph{subgraph-centric} model~\cite{tian2013think}. 

The vertex-centric model, labeled as ``think like a vertex", is an engineering approach similar in concept to MapReduce~\cite{dean2008mapreduce}.
Although the vertex-centric model has been successfully used in many parallel graph computation applications, researchers often report that it has significant bottlenecks in the communication for exchanging excessive messages through the network~\cite{tian2013think}.
On the other hand, the subgraph-centric model (a.k.a. block-centric, partition-centric)~\cite{fan2017parallelizing,tian2013think,yan2014blogel} provides another approach for the parallel graph computation, and has been applied in some state-of-the-art platforms~\cite{fan2017parallelizing}. 
Compared to the vertex-centric model, the subgraph-centric model focuses on the entire subgraph.
This model is labeled as ``think like a graph".
Since subgraphs are more ``coarsely grained'' than a single vertex, they retain many inner edges. 
For this reason, many messages are omitted from transferring through the network~\cite{fan2017parallelizing,yan2014blogel}.
Thus, the subgraph-centric model usually has less communication and converges faster. 

Although the subgraph-centric model has many advantages, it does not reach its full potential in processing large-scale real-world graphs such as Twitter, Facebook and citation graphs~\cite{gonzalez2012powergraph}.
These real-word graphs share a common property that their vertex degree distributions are long tailed. Therefore, they are categorized as power-law graphs~\cite{albert2002statistical}.
Existing graph partition algorithms have difficulties in partitioning power-law graphs.
For example, METIS~\cite{karypis1997parmetis} and NE~\cite{ne} aim to minimize the number of edges and vertices cut by partitioning, hence reducing the cost of communication and improving the efficiency of graph partition results.
However, METIS only considers the number of vertices and NE only considers the number of edges when balancing their results.
Since the degree distribution of power-law graphs is skewed, the number of edges incident on each vertex varies greatly.
Thus they can not balance both vertices and edges by evenly distributing one of them for power-law graphs.

To improve the performance of subgraph-centric frameworks on large-scale power-law graphs, we analyze the communication pattern of several different partition algorithms. 
Based on our analysis, we propose an efficient and balanced vertex-cut graph partition algorithm (EBV).
EBV assigns each edge based on the current value of a well-designed evaluation function, which considers both the total number of cutting vertices and the balance of graph partition results. 
For handling the skewed degree distribution of power-law graphs, we adopt appropriate weights to balance the number of edges and vertices when partitioning.
Moreover, we design an edge sorting preprocessing process, which sorts edges in ascending order by the sum of their end-vertices' degrees before partitioning.
For power-law graphs, the edges with two low-degree end-vertices are assigned at the beginning.
Since the degrees of their end-vertices are low, they are less likely to share the same end-vertex.
Thus the balance of graph partition results is the major factor in the early stage and these low-degree vertices are evenly assigned to each subgraphs as seeds.
We analyze the effects of the evaluation function and the edge sorting preprocessing process theoretically and experimentally.
Our experiments show that EBV outperforms the state-of-the-art partition algorithms for large-scale power-law graphs. 

This paper makes the following contributions:
\begin{itemize}
	\item We propose an efficient and balanced vertex-cut graph partition algorithm.
	\item We compare EBV with the state-of-the-art graph partition algorithms and several parallel graph computation frameworks in detail.
	\item We use the number of communication messages as a platform-independent metrics to compare partition algorithms.
	\item We study the influence of the sorting preprocessing for the EBV algorithm.
\end{itemize}

This paper is organized as follows:
In Section~\ref{sec:motivation} we discuss the challenges arising from large-scale power-law graphs and explain our motivation.
In Section~\ref{sec:back} we review the current graph partition algorithms and introduce the basic notations and metrics for this paper.
In Section~\ref{sec:implementation} we propose an efficient and balanced vertex-cut graph partition algorithm on the subgraph-centric model. 
In Section~\ref{sec:experiments} we present and analyze several experiments, which demonstrate the characteristics and capabilities of our algorithm. 
Finally, we discuss the related work, conclude our work and preview some future projects in Section~\ref{sec:relatedwork} and Section~\ref{sec:conclusion}.




%% file: motivation.tex
\section{Motivation}\label{sec:motivation}
Considering the scope of information used when partitioning, current graph partition algorithms can be classified into two major categories: local-based and self-based.
When partitioning graphs, the local-based algorithms assign edges and vertices according to the local information (part of the graph)~\cite{local}.
However, the self-based algorithms only consider their own attributes (e.g. vertex degree, ID) when assigning edges and vertices.

Subgraph-centric frameworks usually employ some local-based frameworks such as METIS~\cite{karypis1997parmetis,wen2018drone} and its variants~\cite{tian2013think}. 
Neighbor Expand (NE)~\cite{ne} is also a local-based partition algorithm, which expands the new ``core vertex" by searching in the boundary set.
The local-based partition algorithms (METIS and NE) are more concerned about reducing the number of replicated edges and vertices while saving the local structure.
However, the skewed degree distribution of power-law graphs means that the ratios of edges to vertices in the local structures are not uniform, which brings difficulties in balancing both edges and vertices while keeping the local structures.
The imbalanced assignment further results in the imbalanced communication and computation of graph algorithms.

On the other hand, vertex-centric frameworks often use self-based partition algorithms, such as Giraph~\cite{Giraph}, Powergraph~\cite{gonzalez2012powergraph} and Galois~\cite{Galois}.
Recently, many self-based graph partition algorithms have also been proposed for power-law graphs.
The Degree-Based Hashing~\cite{xie2014distributed} (DBH) makes effective use of the skewed degree distribution of power-law graphs. 
It assigns each edge by hashing the ID of its end-vertex with a lower degree.
PowerLyra~\cite{chen2019powerlyra} adopts the hybrid-cut that distinguishes the low-degree and high-degree vertices during the partition process.
Inspired by Fennel~\cite{Fennel}, which is a greedy streaming edge-cut framework, they further propose Ginger by improving the hybrid-cut.
These partition algorithms are simple and efficient to handle large-scale power-law graphs and produce the good balance property.
However, they ignore and destroy the local structure of partitioned graphs.
Therefore, the total communication volume of them is much larger than those local-based algorithms.

Both the total communication volume and the message imbalance need to be considered to find an optimal partition result in power-law graphs.
We demonstrate this principle with a comparison of METIS, NE, Ginger, DBH and CVC.
We also seek to devise a novel partition algorithm following the above principle.
However, the optimal algorithm considering both the total communication volume and the message imbalance has been proved to be NP-hard~\cite{bui1992finding}.
Thus we follow the self-based approach and consider these two factors in our algorithm.

For self-based graph partition algorithms, the edge processing (assigning) order greatly affects the performance of the partition results.
We propose a sorting preprocessing mechanism, which sorts edges with the sum of their end-vertices' degrees in ascending order.
We demonstrate the effect of our sorting preprocessing with a comparison of the alphabetical order in Figure~\ref{fig:order_comp}.
In Figure~\ref{fig:order_comp}, the raw graph is an undirected graph with uneven degree distribution.
We partition the raw graph with our EBV algorithm in different edge processing orders.
The result of sorting preprocessing is more balanced than that of the alphabetical order.

For the result of alphabetical order, it assigns edge $(B, C)$ at the end.
Considering the balance of partition results, $(B, C)$ should be assigned to subgraph $0$.
However, it will lead to cut two more vertices $B$ and $C$.
Thus $(B, C)$ is assigned to subgraph $1$ by EBV.

\begin{figure}
	\centering
	\includegraphics[width=0.9\linewidth]{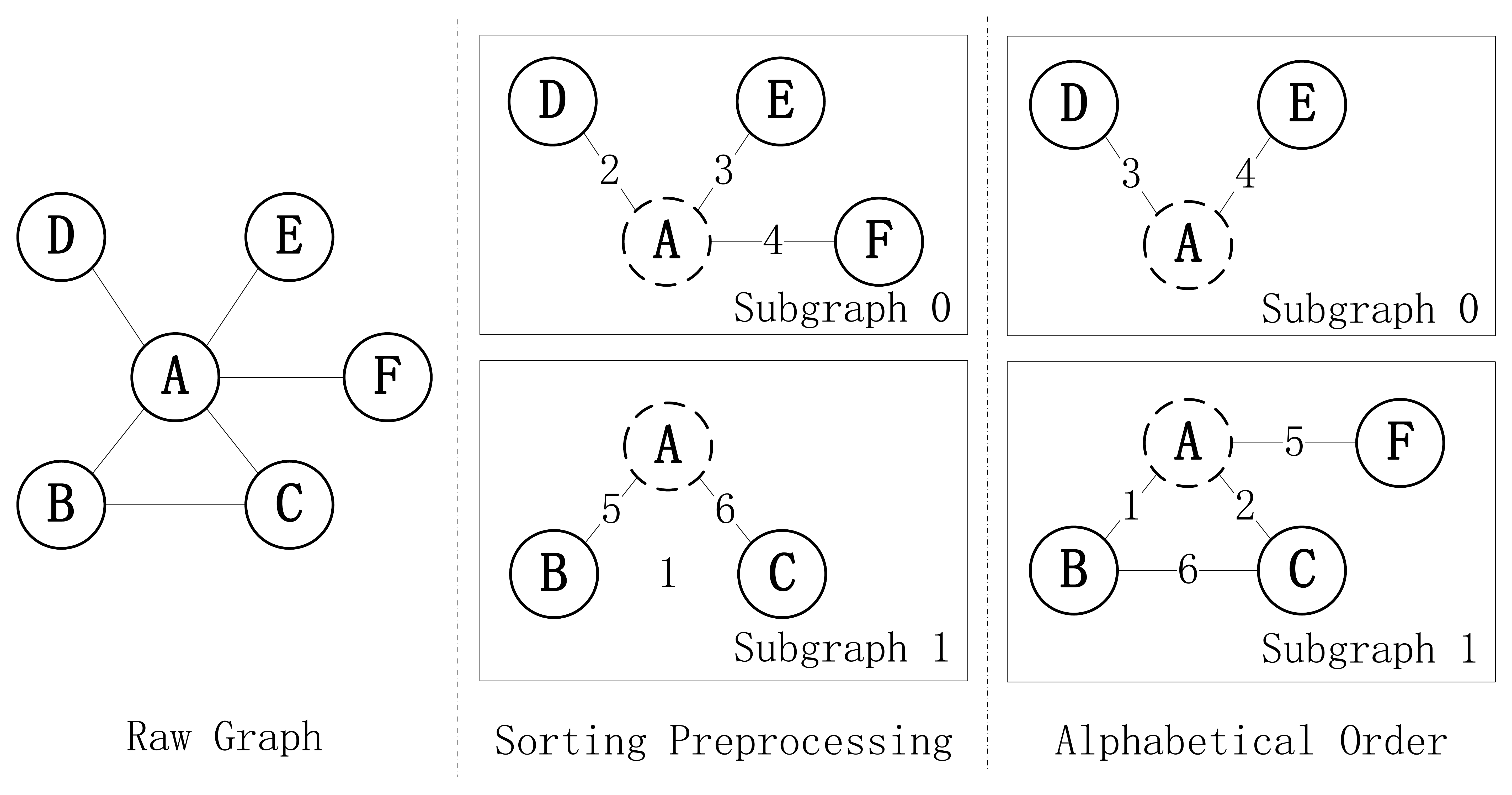}
	\caption{Edge processing order comparison \protect\\
		A comparison of partitioning the raw graph in sorting preprocessing order and alphabetical order. 
		The cut vertices are represented as dot line and the number on each edge indicates the processing order.
	}
	\label{fig:order_comp}
\end{figure}


%% file: background.tex
\section{Background and Preliminaries}\label{sec:back}
In this section, we introduce the properties of large-scale power-law graphs and contemporary graph partition algorithms. We also introduce some notations and metrics for the rest of this paper.

\subsection{Large-Scale Power-law Graph}\label{sec:powerlaw}
Processing large-scale real-world graphs is one of the key topics in the parallel graph computation. One of the most notable properties for real-world graphs is the power-law degree distribution~\cite{gonzalez2012powergraph}.
The power-law graph is also called the scale-free graph. 
The power-law degree distribution is common for natural graphs such as social networks (Twitter and collaboration networks) and the graph of the World Wide Web. 
In the power-law graph, the probability of degree $d$ for a randomly sampled vertex is given by: 
\begin{equation}\label{equ:power} \mathbb{P}(degree = d) \propto d^{-\eta} \end{equation}
where the exponent $\eta$ is a positive constant. 
The lower $\eta$ is, the more skewed the graph will be. 

\subsection{Graph Partitioners}
Existing graph partition algorithms can be divided into two main categories: edge-cut (vertex partitioning) and vertex-cut (edge partitioning).
Edge-cut method cuts the cross-partition edges and balances the number of vertices in all subgraphs.
Each subgraph needs to maintain the routing messages and the end-vertices on the other side (ghost vertices) of those replicated edges.
On the contrary, vertex-cut method cuts the vertices and balances the number of edges.
The cut vertices (replicated vertices) are also maintained in multiply subgraphs.
The majority of distributed graph computing systems (e.g.~\cite{low2012distributed, malewicz2010pregel, tian2013think, yan2014blogel}) use edge-cut for graph partitioning. 
However, vertex-cut has been proposed in~\cite{gonzalez2012powergraph} as a better approach to process graphs with power-law degree distributions. 
Thus many current graph partition algorithms are based on the vertex-cut method, such as CVC~\cite{boman2013scalable}, Ginger~\cite{chen2019powerlyra}, NE~\cite{ne} and ADWISE~\cite{ADWISE}.
We also adopt the vertex-cut method for our EBV algorithm.

\subsection{Preliminaries}\label{sec:pre}
Let $G = (V, E)$ be a directed graph, where $V$ denotes the set of vertices and $E$ denotes the set of edges. 
An edge is represented as a pair of vertices $(u, v)$ where $u$ is the source vertex and $v$ is the target vertex.
For the undirected graph, we replace each undirected edge with two edges with opposite directions.
Suppose we partition the graph $G$ into $p$ subgraphs, we denote $G_i(V_i, E_i)$, $i \in [1, p]$ as the $i^{th}$ subgraph.

The vertex-cut (edge partitioning) algorithms partition the edge set $E$ into $p$ subsets.
Let $E_1 \cup E_2 \cup \cdots \cup E_p = E$ be the $p$ partitions of $E$, i.e. $E_i \cap E_j = \emptyset$, $\forall i \ne j$. 
Here we define $V_i = \{u \mid (u, v) \in E_i \lor (v, u) \in E_i \}$ as the set of vertices covered by $E_i$. 
Due to the definition of $V_i$, there may exist some replicated vertices which have the same vertex ID but belong to different subgraphs.

For the edge-cut (vertex partitioning) algorithms, they partition the vertex set $V$.
Let $V_1 \cup V_2 \cup \cdots \cup V_p = V$ and $V_i \cap V_j = \emptyset$ for all $i \ne j$, we define $E_i = \{(u, v) \mid u \in V_i \lor v \in V_i \}$. The definition of $E_i$ here means that there exists some replicated edges such that $E_i \cap E_j \ne \emptyset$.

Further, we introduce three metrics: edge imbalance factor, vertex imbalance factor and replication factor.
These metrics are also widely used in~\cite{boman2013scalable,chen2019powerlyra,xie2014distributed, ne}.
The edge imbalance factor is defined as  
\begin{math}
\frac{\max_{i=1,...,p} |E_{i}|}{|E|/p}
\end{math} 
, while the vertex imbalance factor is defined as 
\begin{math}
\frac{\max_{i=1,...,p} |V_{i}|}{\sum_{i=1}^{p} |V_{i}|/p}
\end{math}
. Both of them are used to measure the balance of partition results.
The replication factor for the vertex-cut algorithm is defined as   
\begin{math}
\frac{\sum_{i=1}^{p} |V_{i}|}{|V|}
\end{math}.
However, we have $\sum_{i=1}^{p} |V_{i}| = |V|$ for the edge-cut algorithm.
Thus the definition of the replication factor cannot be directly adapted and we define 
\begin{math}
\frac{\sum_{i=1}^{p} |E_{i}|}{|E|}
\end{math} as the replication factor for the edge-cut algorithm.
The replication factor represents the average number of replicas for a vertex or edge.

%% file: droneapi.tex
\section{Efficient and Balanced Vertex-Cut Partition Algorithm}\label{sec:implementation}
\subsection{Overview}
In order to solve the large-scale power-law graph computation problem in subgraph-centric frameworks, we devise a highly scalable efficient and balanced vertex-cut partition algorithm (EBV).
First, we demonstrate the subgraph-centric, bulk synchronous parallel programming model and workflow as our main test case.
Second, we propose and analyze our EBV algorithm.
Finally, we prove the upper bound of the edge imbalance factor and vertex imbalance factor of the EBV algorithm for general graphs.

\subsection{Subgraph-Centric Bulk Synchronous Parallel Model and Workflow}\label{sec:api}
The subgraph-centric, bulk synchronous parallel (BSP) model~\cite{valiant1990bridging} is one of the most popular programming models for the subgraph-centric approach. 
All of the partition algorithms we test are based on this model.
In this model, the whole graph is divided into several subgraphs. 
Each subgraph is bound to one worker (process), and each worker handles only one subgraph. 
For fair comparison, here we do not use multi-threading technology for accelerating. 

The entire graph processing of subgraph-centric, BSP model is iterative and can be divided into supersteps with three stages: the computation stage (update graph), the communication stage (exchange messages) and the synchronization stage (wait for other workers to complete message exchanging).
In the computation stage, the sequential algorithm takes the current subgraph and receiving messages as input. 
Each subgraph is updated by the sequential algorithm. 
In the communication stage, only the message sending/receiving operations among the replicated vertices are allowed.
Usually, the messages contain sufficient information for updating the receiving vertices' states.
The synchronization stage is designed as a barrier for separating supersteps.
In this stage, each worker waits for other workers to finish their computation and communication stages.
The term ``bulk synchronous'' is derived from such a process.


\subsection{Efficient and Balanced Vertex-Cut Partition}\label{sec:greedy}
\begin{algorithm}[htb]
	\caption{Efficient and Balanced Vertex-Cut Partition Algorithm}
	\label{alg:greedy}
	\DontPrintSemicolon
	\KwIn{Graph $G(V,E)$, the number of subgraphs $p$}
	\KwOut{Partition result $part$, $part_{(u,v)}$ is the part assignment for edge $(u,v)$}
	
	\For{$i \in [1, p]$} {
		$//$ $keep[i]$ saves the vertex set that the $i^{th}$ subgraph should keep.\;
		$keep[i] \gets \emptyset$\;
		$//$ $e_{count}[i]$ and $v_{count}[i]$ are the number of edges and vertices of the $i^{th}$ subgraph.\;
		$e_{count}[i] \gets 0$,  $v_{count}[i] \gets 0$ \;
	}
    \For{$(u, v) \in E$} {
    	$//$ Calculate the evaluation function\;
    	\For{$i \in [1, p]$} {
    		$Eva[i] \gets 0$\;
    		\If{$u \notin keep[i]$} {
    			$Eva[i] \gets Eva[i] + 1$\;
    		}
    		\If{$v \notin keep[i]$} {
    			$Eva[i] \gets Eva[i] + 1$\;
    		}
    		$Eva[i] \gets Eva[i] + \alpha \frac{e_{count}[i]}{|E| / p} + \beta \frac{v_{count}[i]}{|V|/p}$
    	}
    	$part_{(u,v)} \gets \mathop{\arg\min}\limits_{ i} Eva[i]$\;
    	$//$ Update variables for further partition\;
    	$e_{count}[part_{(u,v)}] \gets e_{count}[part_{(u,v)}] + 1$\;
    	\If{$u \notin keep[part_{(u,v)}]$}{
    		$v_{count}[part_{(u,v)}] \gets v_{count}[part_{(u,v)}] + 1$\;
    	}
   		 \If{$v \notin keep[part_{(u,v)}]$}{
    		$v_{count}[part_{(u,v)}] \gets v_{count}[part_{(u,v)}] + 1$\;
    	}
    	$keep[part_{(u,v)}] \gets keep[part_{(u,v)}] \cup \{u, v\}$\;
    }
\end{algorithm}

Algorithm~\ref{alg:greedy} describes in details how our EBV partitions the whole graph $G(V, E)$ into $p$ subgraphs. 
The EBV algorithm takes the graph $G(V, E)$ and the number of subgraphs $p$ as input and outputs the partition result.
The $keep$, $e_{count}$ and $v_{count}$ are auxiliary variables and updated dynamically. 
$keep[i]$ saves the vertex set of the $i^{th}$ subgraph, while $e_{count}[i]$ and $v_{count}[i]$ represent the number of edges and vertices have been assigned to the $i^{th}$ subgraph.

We abstract an evaluation function
\begin{equation} 
\label{equ:eva}
\begin{aligned}
Eva_{(u, v)}(i) =& \mathbb{I}(u \notin keep[i]) + \mathbb{I}(v \notin keep[i]) \\
+& \alpha \frac{e_{count}[i]}{|E| / p} + \beta \frac{v_{count}[i]}{|V|/p}
\end{aligned}
\end{equation}
from this algorithm.
The evaluation function $Eva_{(u, v)}(i)$ is used to measure the benefit of assigning edge $(u, v)$ to subgraph $i$.
The smaller $Eva_{(u, v)}(i)$ is, the more suitable this assignment is.
Our algorithm will assign edge $(u, v)$ to worker $i$ which has the smallest $Eva_{(u, v)}(i)$.
The $\mathbb{I}(State)$ in (\ref{equ:eva}) is an indicator function, which returns $1$ for $State$ is true and $0$ for false.
The hyperparameters $\alpha$ and $\beta$ are used to adjust sensitivity for the balance of edges and vertices in the evaluation function.
The larger they are, the more our algorithm focuses on balancing edges and vertices.
For the experiments in Section~\ref{sec:experiments}, we set $1$ as the default value of $\alpha$ and $\beta$.

The evaluation function can be divided into three parts: $\mathbb{I}(u \notin keep[i]) + \mathbb{I}(v \notin keep[i])$, $\alpha \frac{e_{count}[i]}{|E| / p}$ and $\beta \frac{v_{count}[i]}{|V|/p}$. 
$\mathbb{I}(u \notin keep[i]) + \mathbb{I}(v \notin keep[i])$ is related to the replication factor, while $\alpha \frac{e_{count}[i]}{|E| / p}$ and $\beta \frac{v_{count}[i]}{|V|/p}$ restrict the edge and vertex imbalance factor respectively. 
Thus, EBV can balance the overall communication/computation overhead and workload imbalance, and outperforms traditional graph partition algorithms.

Moreover, as a sequential graph partition algorithm, the quality of results for EBV is naturally affected by the edge processing order.
For offline partition jobs, we sort edges in ascending order by the sum of end-vertices' degrees before the execution of EBV.
Intuitively, this sorting mechanism prevents edges with high-degree end-vertices from being assigned at the early stage.
Since edges with two low-degree vertices are unlikely to share the same end-vertex, they are mainly assigned based on the imbalance constraints: $\alpha \frac{e_{count}[i]}{|E| / p}$ and $\beta \frac{v_{count}[i]}{|V|/p}$.
Thus, they tend to be evenly assigned to each subgraph at the beginning of EBV.
Generally, each subgraph keeps low-degree vertices as seeds within themselves in the early stage and cut few high-degree vertices later.
We verify this hypothesis through experiments in Section~\ref{sec:sorting}.

\subsection{Upper Bound of Edge and Vertex Imbalance Factors}\label{sec:alpha}
In this section, we show that the worst-case upper bound of the edge imbalance factor and the vertex imbalance factor of EBV for general graphs are $1 + \frac{p-1}{|E|} \left(1 + \left\lfloor \frac{2|E|}{\alpha p} + \frac{\beta}{\alpha} |E| \right\rfloor \right)$ and $1 + \frac{p-1}{\sum_{j = 1}^p|V_j|} \left(1 + \left\lfloor \frac{2|V|}{\beta p}  + \frac{\alpha}{\beta} |V| \right\rfloor \right)$ respectively. This result means that we can restrict the upper bound of edge and vertex imbalance factors by tuning the hyperparameters $\alpha$ and $\beta$.

\begin{theorem}
	Given any graph $G(V, E)$ and any positive integer $p$, partitioning $G$ into $p$ subgraphs by Algorithm~\ref{alg:greedy}, the upper bound of the edge imbalance factor for the partition result is $1 + \frac{p-1}{|E|} \left(1 + \left\lfloor \frac{2|E|}{\alpha p} + \frac{\beta}{\alpha} |E| \right\rfloor \right)$.
\end{theorem}

\begin{proof}
	For the sake of simplicity, we denote $Eva^m_{(u, v)}(i)$, $e_{count}^m[i]$, $v_{count}^m[i]$ and $keep^m[i]$ as the value of $Eva_{(u, v)}(i)$, $e_{count}[i]$, $v_{count}[i]$ and $keep[i]$ before assigning the $m^{th}$ edge, while $m \in [1, |E|]$.
	Specifically, we denote $e_{count}^{|E| + 1}[i]$ as $|E_i|$, the number of edges of the final subgraph $i$.
	
	Let the $m^{th}$ edge $(u, v)$ be assigned to subgraph $i$ ($1 \le i \le p$) by Algorithm~\ref{alg:greedy}.
	By line $15$ of Algorithm~\ref{alg:greedy}, we have
	\begin{equation} 
	\label{equ:inequality_0}
	Eva^m_{(u,v)}(i) - Eva^m_{(u,v)}(j) \le 0
	\end{equation}
	for any $j \in [1, p], i \ne j$.
	
	Substitute (\ref{equ:eva}) to (\ref{equ:inequality_0}), we obtain 
	\begin{equation} 
	\label{equ:inequality}
	\begin{aligned}
	\alpha \frac{e_{count}^m[i] - e_{count}^m[j]}{|E| / p} \le& \beta \frac{v_{count}^m[j] - v_{count}^m[i]}{|V| / p} \\
	+ &\mathbb{I}(u \notin keep^m[j]) +\mathbb{I}(v \notin keep^m[j]) \\
	- &\mathbb{I}(u \notin keep^m[i]) - \mathbb{I}(v \notin keep^m[i]).
	\end{aligned}
	\end{equation}
	
	Indicator function $\mathbb{I}(State)$ is a boolean function. Therefore, the upper bound of $\mathbb{I}(u \notin keep^m[j]) + \mathbb{I}(v \notin keep^m[j]) - \mathbb{I}(u \notin keep^m[i]) - \mathbb{I}(v \notin keep^m[i])$ is $2$.
	As we have $0 \le v_{count}^m[i], v_{count}^m[j] \le |V|$, $|v_{count}^m[i] - v_{count}^m[j]| \le |V|$, we can apply these inequalities to (\ref{equ:inequality}) to derive
	\begin{equation} 
	\label{equ:inequality_2}
	\begin{aligned}
	e_{count}^m[i] - e_{count}^m[j] \le \frac{2|E|}{\alpha p} + \frac{\beta}{\alpha} |E| \quad.
	\end{aligned}
	\end{equation}
	
	Since $e_{count}^m[i]$ and $e_{count}^m[j]$ are integers, (\ref{equ:inequality_2}) can be rewritten as 
	\begin{equation} 
	\label{equ:inequality_3}
	\begin{aligned}
	e_{count}^m[i] - e_{count}^m[j] \le \left\lfloor \frac{2|E|}{\alpha p}  + \frac{\beta}{\alpha} |E| \right\rfloor \quad.
	\end{aligned}
	\end{equation}
	
	To complete the proof, we rely on the following technical Lemma~\ref{lemma}.
	\begin{lemma}\label{lemma}
		For any graph $G(V, E)$ partitioned into $p$ subgraphs by Algorithm~\ref{alg:greedy}, $e_{count}^m[i] - e_{count}^m[j] \le 1 + \left\lfloor \frac{2|E|}{\alpha p}  + \frac{\beta}{\alpha} |E| \right\rfloor$ holds for any integer $i$, $j$ and $m$ satisfying $1 \le i \ne j \le p$ and $m \in [1, |E| + 1]$.
	\end{lemma}
	
	\begin{proof}
		For any $i, j \in [1, p], i \ne j$ and $m \in [1, |E|]$ such that $e_{count}^m[i] - e_{count}^m[j] > \left\lfloor \frac{2|E|}{\alpha p}  + \frac{\beta}{\alpha} |E| \right\rfloor$.
		Inequality (\ref{equ:inequality_3}) indicates that the $m^{th}$ edge will not be assigned to subgraph $i$.
		Therefore, 
		\begin{equation} 
		e_{count}^{m+1}[i] - e_{count}^{m+1}[j] \le e_{count}^m[i] - e_{count}^m[j]
		\end{equation}
		when $e_{count}^m[i] - e_{count}^m[j] = 1 + \left\lfloor \frac{2|E|}{\alpha p} + \frac{\beta}{\alpha} |E| \right\rfloor$.
		
		Besides, for any $i \in [1, p]$, $e_{count}^1[i] = 0$.
		
		Thus the lemma can be proved by mathematical induction. For the sake of brevity, we omit the details.
	\end{proof}
	
	Since $e_{count}^{|E| + 1}[i] = |E_i|$ and $\sum_{i = 1}^p |E_i| = |E|$,
	\begin{equation} 
	\label{equ:prove_1}
	\begin{aligned}
	\sum_{j = 1}^p (|E_i| - |E_j|) = & p \times |E_i| - |E|
	\end{aligned}
	\end{equation}
	for any $i$.
	
	By Lemma 1,
	\begin{equation} 
	\label{equ:prove_2}
	\begin{aligned}
	\sum_{j = 1}^p	(|E_i| - |E_j|)& = \sum_{j = 1, j \ne i}^p (|E_i| - |E_j|) \\
	& \le (p -1) \times \left(1 + \left\lfloor \frac{2|E|}{\alpha p} + \frac{\beta}{\alpha} |E| \right\rfloor \right).
	\end{aligned}
	\end{equation}
	
	Substitute (\ref{equ:prove_1}) to (\ref{equ:prove_2}),
	\begin{equation} 
	\label{equ:prove_3}
	\begin{aligned}
	\frac{|E_i|}{|E|/p} \le 1 + \frac{p-1}{|E|} \left(1 + \left\lfloor \frac{2|E|}{\alpha p} + \frac{\beta}{\alpha} |E| \right\rfloor\right)
	\end{aligned}
	\end{equation}
	for any $i \in [1, p]$.
	Thus we have $\frac{\max_{i=1,...,p} |E_{i}|}{|E|/p} \le 1 + \frac{p-1}{|E|} \left(1 + \left\lfloor \frac{2|E|}{\alpha p} + \frac{\beta}{\alpha} |E| \right\rfloor \right) $. 
	
\end{proof}

\begin{theorem}
	Given any graph $G(V, E)$ and any positive integer $p$, partitioning $G$ into $p$ subgraphs by Algorithm~\ref{alg:greedy}, the upper bound of the vertex imbalance factor for the partition result is $1 + \frac{p-1}{\sum_{j = 1}^p|V_j|} \left(1 + \left\lfloor \frac{2|V|}{\beta p}  + \frac{\alpha}{\beta} |V| \right\rfloor \right)$.
\end{theorem}

The proof of Theorem 2 adopts the same method used in the proof of Theorem 1 with minor modifications.
For the sake of simplicity, we do not present it in this paper.

%% file: experiments.tex
\section{Experiments and Analysis}\label{sec:experiments}
In this section, we compare and analyze the performance of EBV with the Ginger, Degree-Based Hash (DBH), Cartesian Vertex-Cut (CVC), Neighbor Expansion (NE) and METIS in large-scale power-law and non-power-law graphs.
Specifically, we address three major research problems:
\begin{enumerate}[(1)]
	\item What's the influence of the total message size and the message imbalance for the parallel graph computation? 
	\item Are there any graph partition metrics that can reflect the total message size and the message imbalance? If we design a graph partition algorithm that aims to optimize both at the same time, is it better than existing algorithms?
	\item Does our sorting preprocessing benefit EBV algorithm?
\end{enumerate}

\subsection{Experiment Setup and Data}
To answer these questions above, we choose three of the most famous graph algorithms as examples: Single Source Shortest Path (SSSP)~\cite{fredman1987fibonacci}, PageRank (PR)~\cite{page1999pagerank} and Connected Component (CC)~\cite{samet1979connected}.
We also choose four large-scale graphs: USARoad~\cite{USARoad}, LiveJournal~\cite{LiveJournal}, Twitter~\cite{Twitter} and Friendster~\cite{Friend}.
The statistics of these graphs are listed in Table~\ref{tab:dataset}. 
Note that LiveJournal, Twitter and Friendster are power-law graphs, while USARoad is not.
The $\eta$ value (Section~\ref{sec:powerlaw}) quantifies the skewed property of power-law graphs.
Since we want to analyze the influence of $\eta$, we provide the $\eta$ value for USARoad according to the definition, although it's not a power-law graph. 

\begin{table*}[htbp]
	\caption{Statistics of tested graphs}
	\label{tab:dataset}
	\centering
	\begin{tabular}{cccccc}
		Graphs &Type & V & E & Average Degree & $\eta$ \\ \midrule
		USARoad & Undirected & $23,947,347$ & $\quad\,\,\,58,333,344$  & $\,\,\,2.44$ & $6.30$ \\ 
		LiveJournal & Directed & $\,\,\,4,847,571$ & $\quad\,\,\,68,993,773$  & $14.23$ & $2.64$ \\ 
		Friendster  & Undirected & $65,608,366$  & $1,806,067,135$  & $27.53$ & $2.43$ \\
		Twitter & Directed & $41,652,230$  & $1,468,365,182$  & $35.25$ & $1.87$ \\
		\bottomrule
	\end{tabular}

\end{table*}


We test six partition algorithms: the EBV we proposed, Ginger, Degree-Based Hash (DBH), Cartesian Vertex-Cut (CVC), Neighbor Expansion (NE) and METIS. 
We refer to their implementation in the subgraph-centric BSP framework DRONE~\cite{wen2018drone} as EBV, Ginger, DBH, CVC, NE and METIS respectively.
Our experiment platform is a 4-node cluster, with each node consisting 8 Intel Xeon E7-8830 2.13GHz CPU and 1TB memory. 

\subsection{Graph Partition Comparison}\label{sec:partitionercompare}
Here we present the execution time comparison of EBV, Ginger, DBH, CVC, NE and METIS. 
We also compare these results with the state-of-the-art frameworks Galois~\cite{Galois} and Blogel~\cite{yan2014blogel}) for diversity. 
In this comparison, the partition overhead is not included.
Figure~\ref{fig:systemcompare} shows the performance comparison.
We should remark that Blogel employs a multi-source based Graph Voronoi Diagram partitioner to detect the connected components in the partitioning phase, ensuring that each block is ``connected''. In the subsequent CC computing phase, it only merges small connected components into a bigger one without performing the vertex-level computation. For a fair comparison, we add the pre-computing time in the partition phase to the total time of CC for Blogel. Blogel is also excluded from the PR comparison because its PR implementation is not standard and its result is not directly comparable. The interested reader is referred to~\cite{kamvar2003exploiting} for more details.

\begin{figure*}
	\centering
	\begin{minipage}{0.32\linewidth}
		\centerline{\includegraphics[width=1\textwidth]{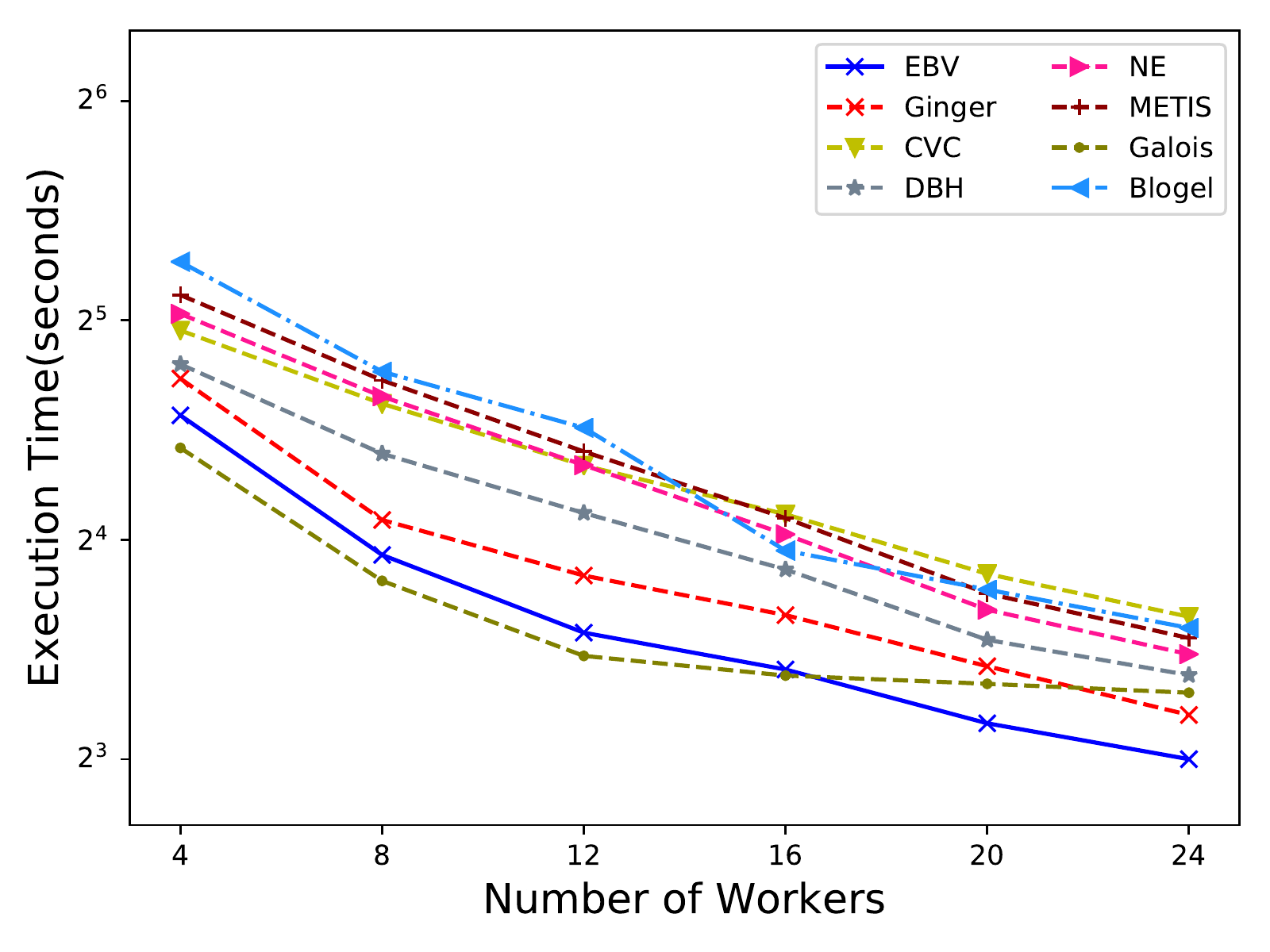}}
		\centerline{CC - LiveJournal}
	\end{minipage}
	\hfill
	\begin{minipage}{0.32\linewidth}
		\centerline{\includegraphics[width=1\textwidth]{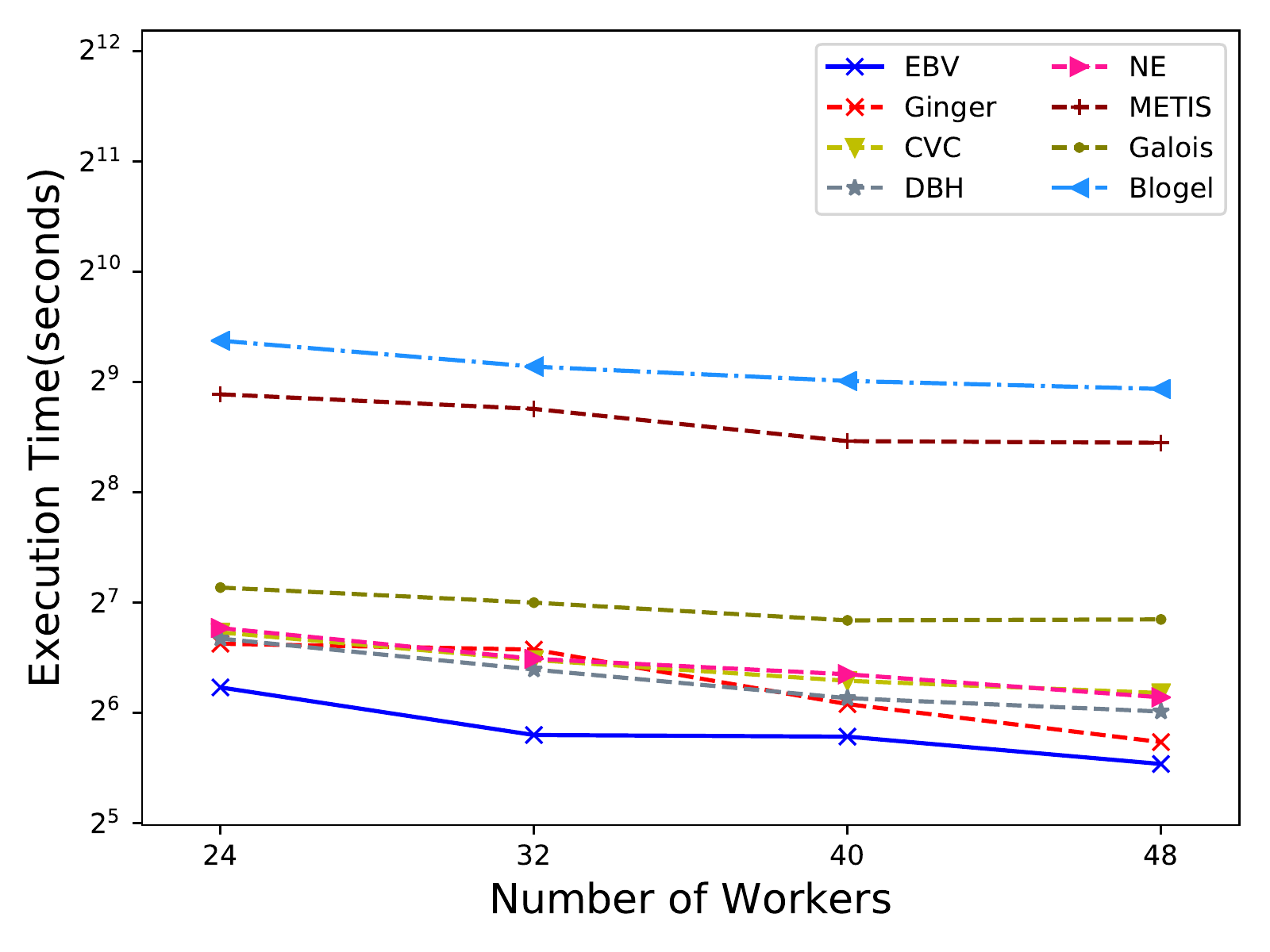}}
		\centerline{CC - Twitter}
	\end{minipage}
	\hfill
	\begin{minipage}{0.32\linewidth}
		\centerline{\includegraphics[width=1\textwidth]{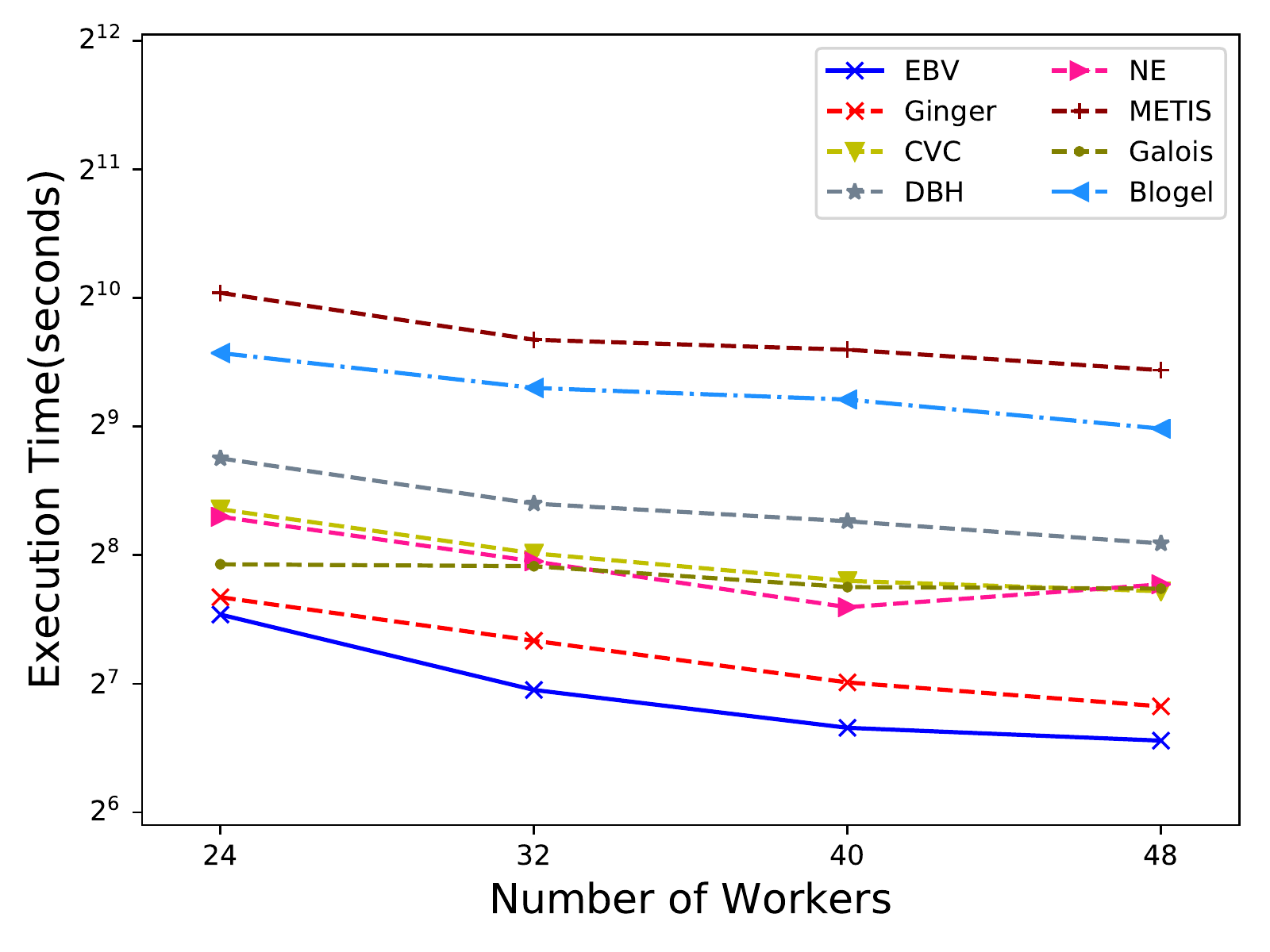}}
		\centerline{CC - Friendster}
	\end{minipage}\\[0.5mm]
	
	\begin{minipage}{0.32\linewidth}
		\centerline{\includegraphics[width=1\textwidth]{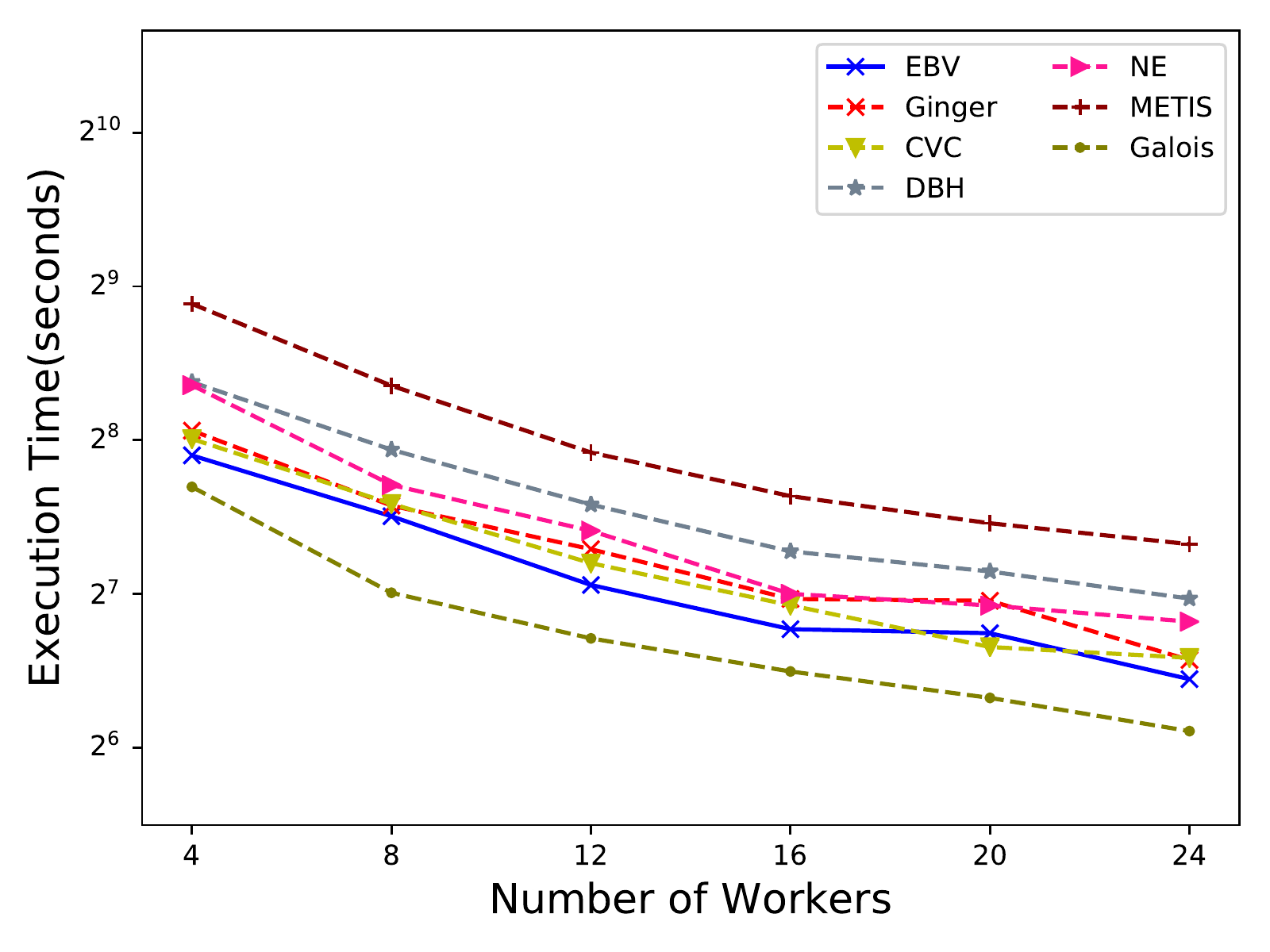}}
		\centerline{PR - LiveJournal}
	\end{minipage}
	\hfill
	\begin{minipage}{0.32\linewidth}
		\centerline{\includegraphics[width=1\textwidth]{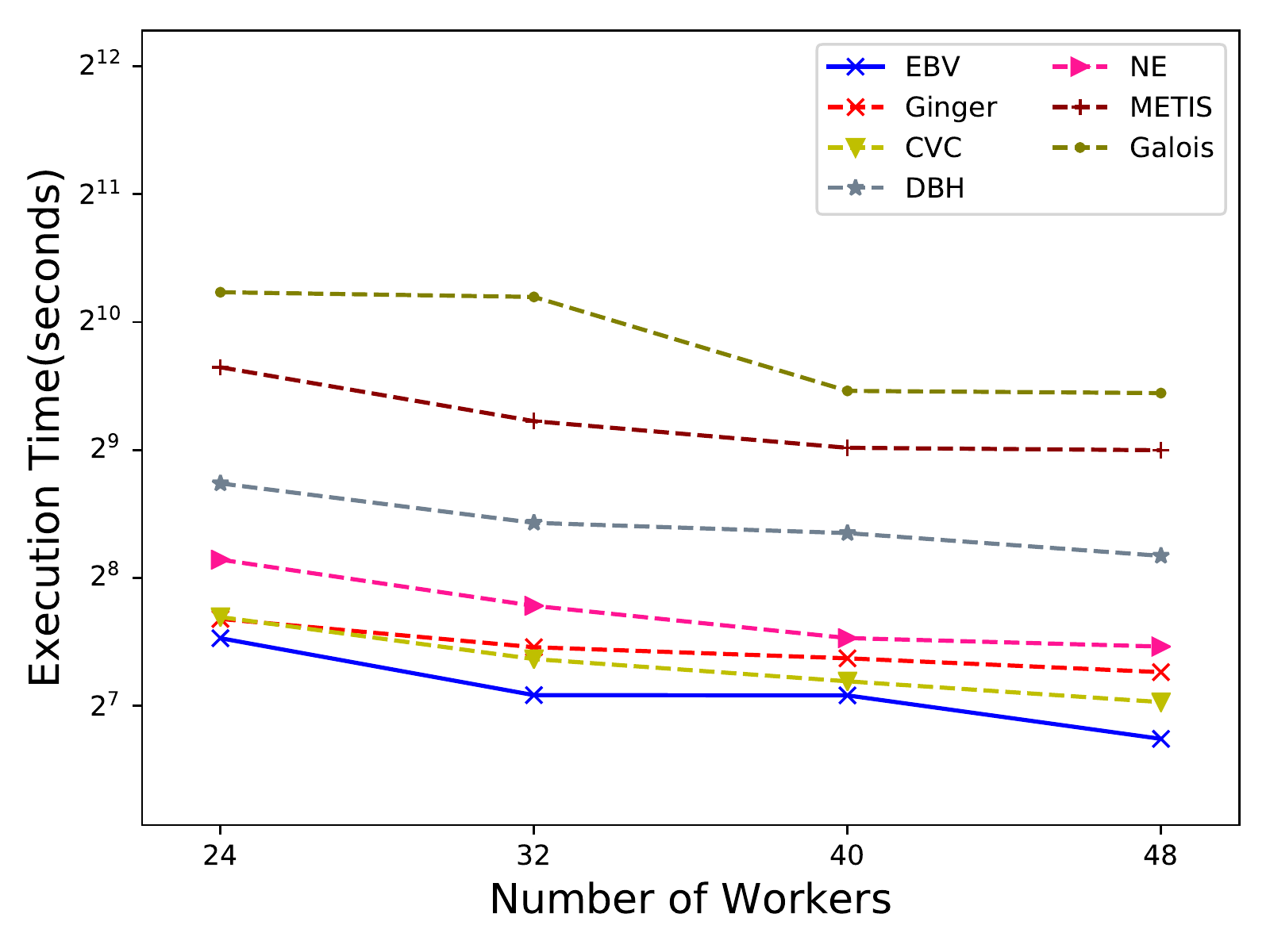}}
		\centerline{PR - Twitter}
	\end{minipage}
	\hfill
	\begin{minipage}{0.32\linewidth}
		\centerline{\includegraphics[width=1\textwidth]{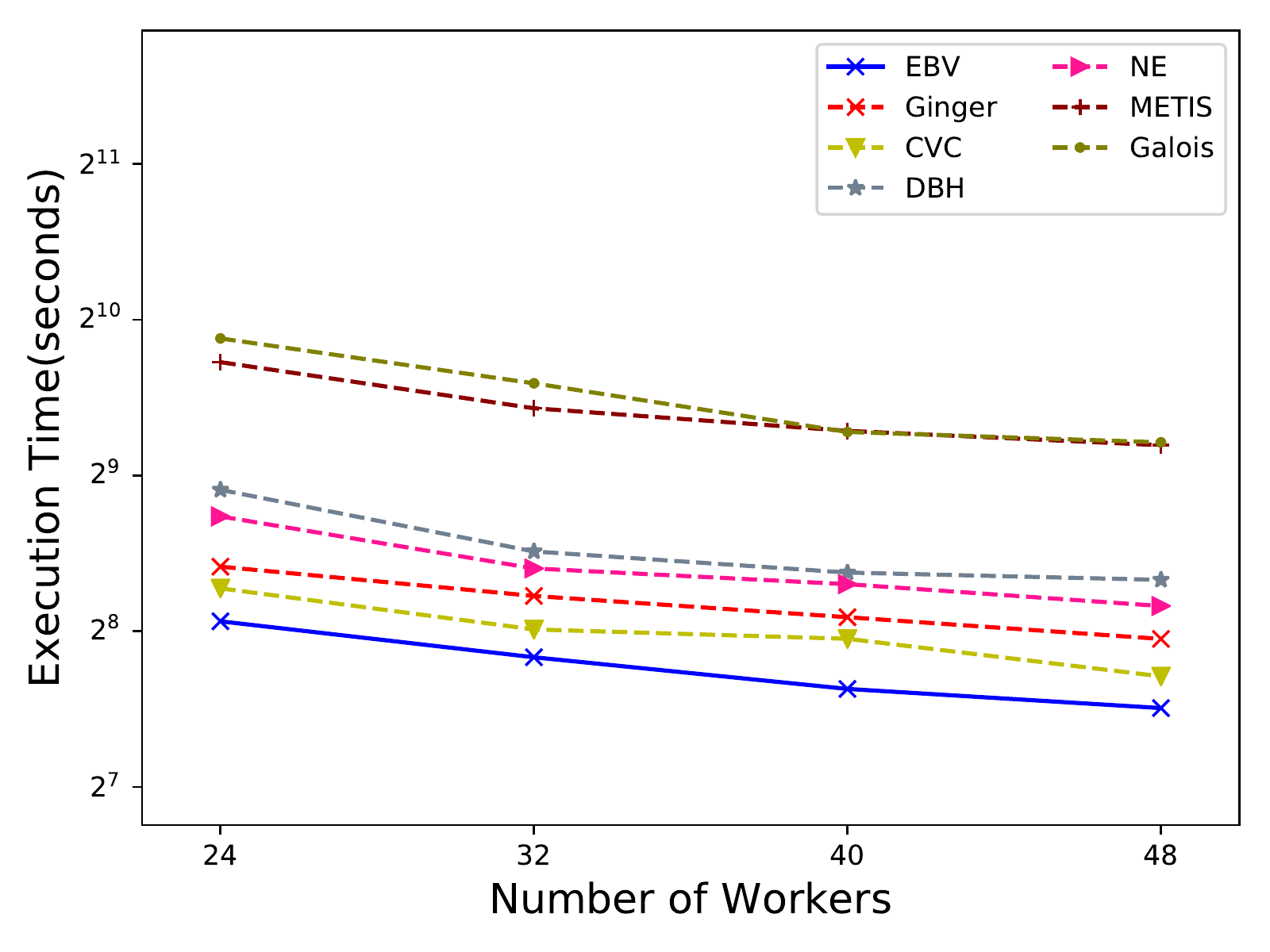}}
		\centerline{PR - Friendster}
	\end{minipage}\\[0.5mm]
	
	\begin{minipage}{0.32\linewidth}
		\centerline{\includegraphics[width=1\textwidth]{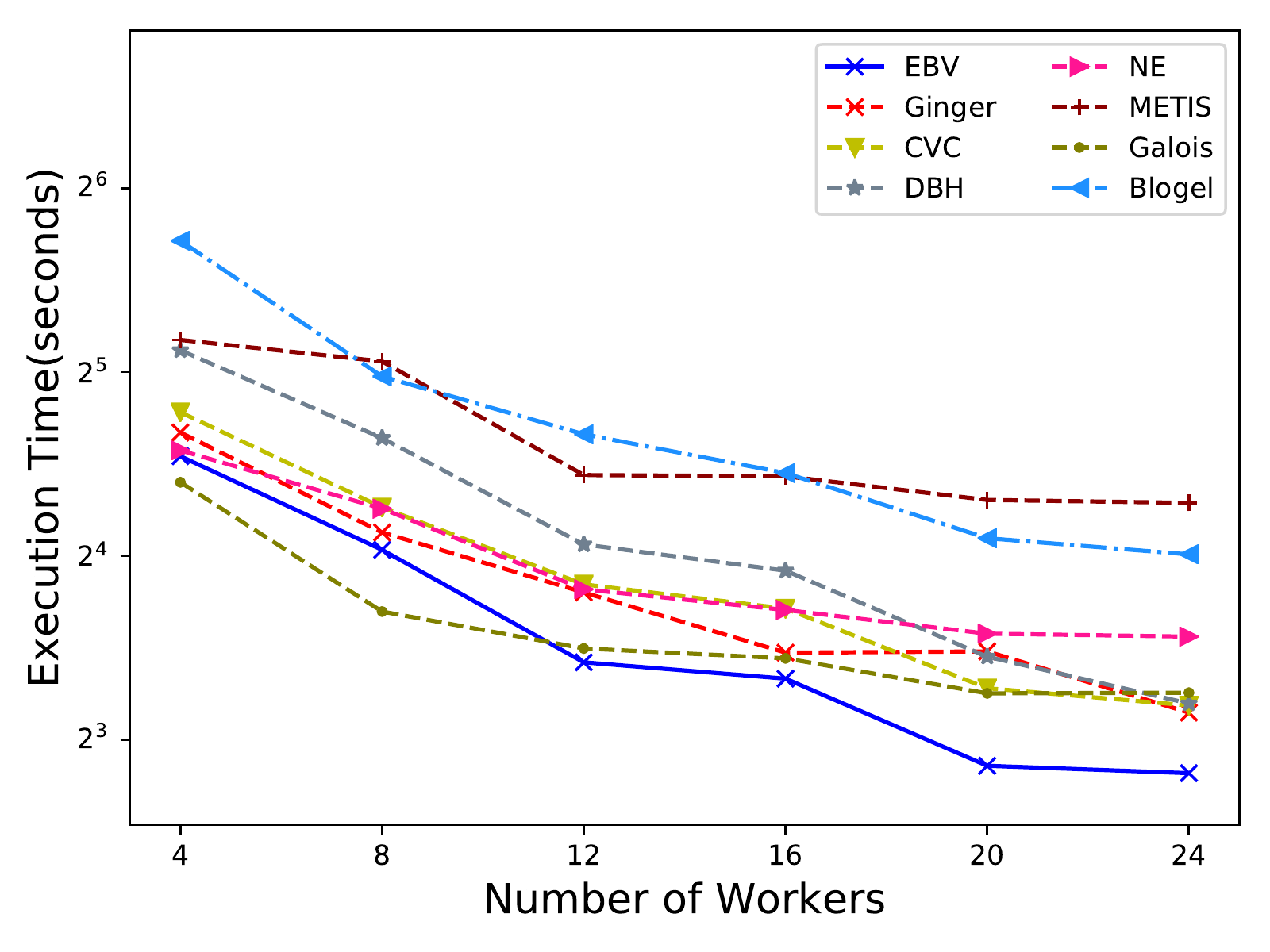}}
		\centerline{SSSP - LiveJournal}
	\end{minipage}
	\hfill
	\begin{minipage}{0.32\linewidth}
		\centerline{\includegraphics[width=1\textwidth]{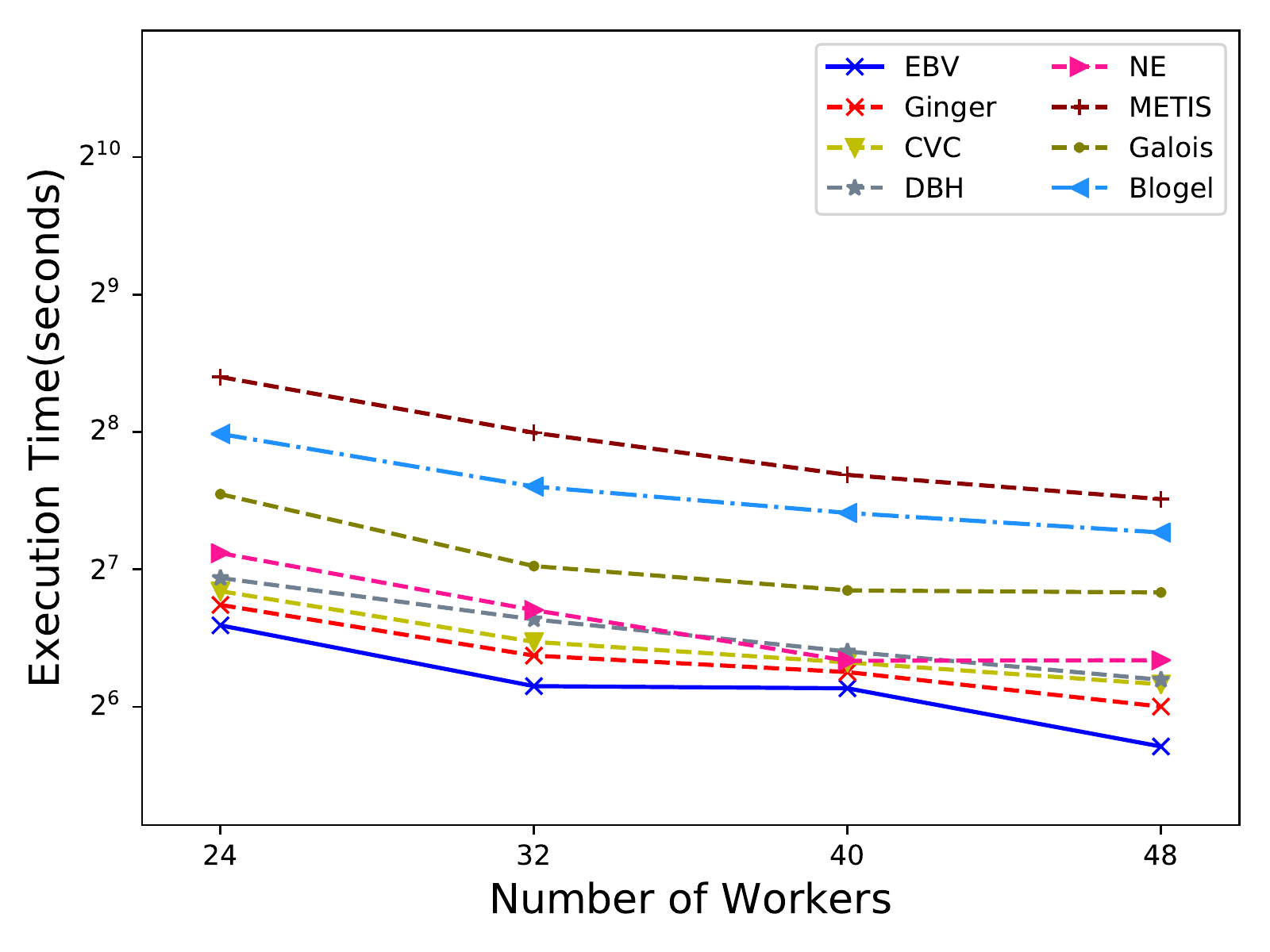}}
		\centerline{SSSP - Twitter}
	\end{minipage}
	\hfill
	\begin{minipage}{0.32\linewidth}
		\centerline{\includegraphics[width=1\textwidth]{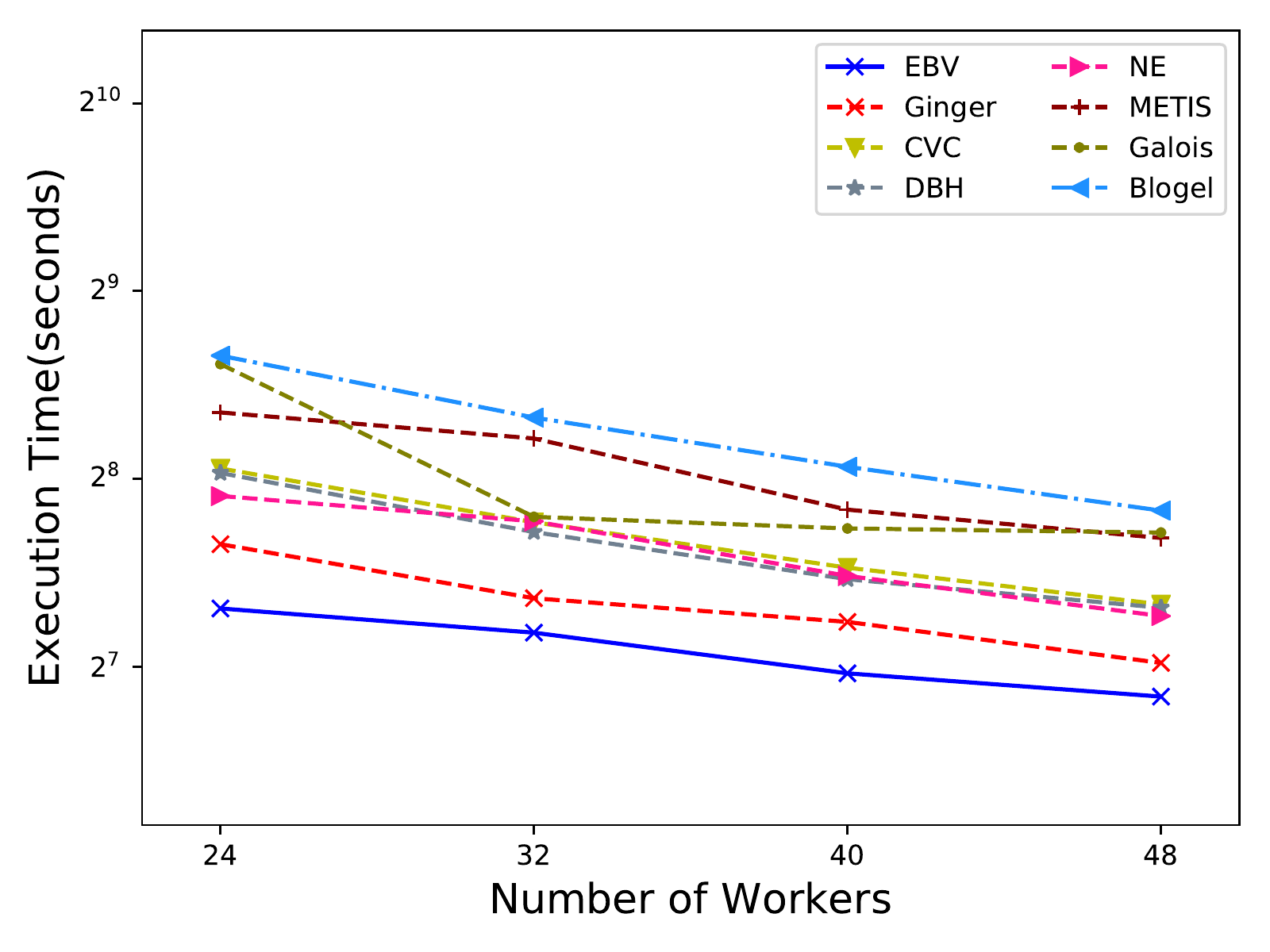}}
		\centerline{SSSP - Friendster}
	\end{minipage}
	
	\caption{
		Cross-system performance comparison of CC, PR and SSSP on various graphs. \protect\\
		Figure~\ref{fig:systemcompare} shows the execution time of CC, PR and SSSP on three power-law graphs with different numbers of workers.
		The X-axis is the number of workers. While the Y-axis is the execution time.
	}
	\label{fig:systemcompare}
\end{figure*}

From Figure~\ref{fig:systemcompare}, we can find that EBV performs best in most cases.
More specifically, compared with Ginger, DBH, CVC, NE and METIS, its execution time is reduced by an average of $16.8\%$, $37.3\%$, $25.4\%$, $31.5\%$ and $63.0\%$ respectively.
Although Galois performs well and even outperforms EBV in the PR-LiveJournal, it is limited for larger graphs.
For Blogel, it takes a longer time in our experiments than Galois.

\begin{figure*}
	\centering
	\begin{minipage}{0.48\linewidth}
		\centerline{\includegraphics[width=1\textwidth]{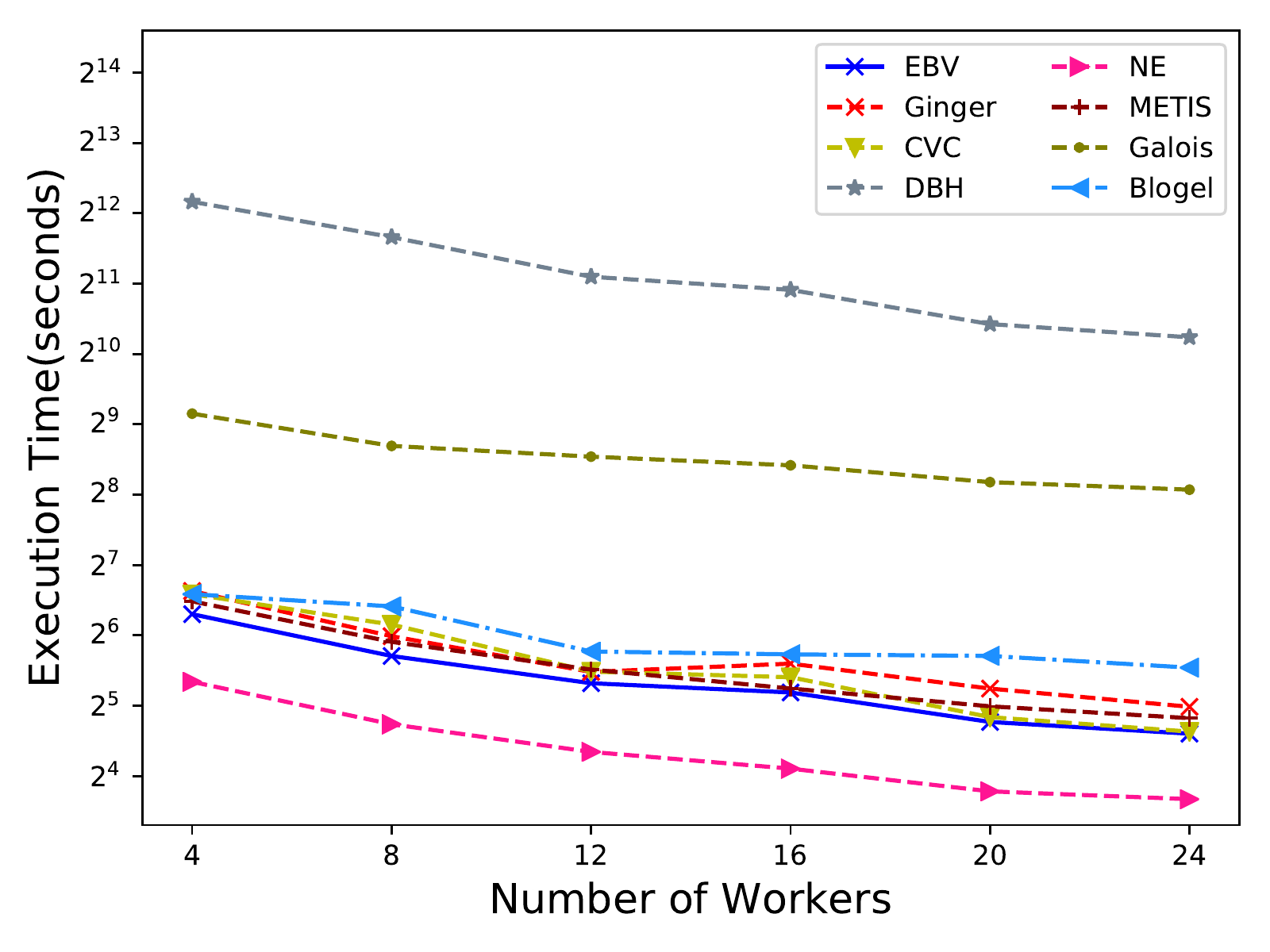}}
		\centerline{CC - USARoad}
	\end{minipage}
	\hfill
	\begin{minipage}{0.48\linewidth}
		\centerline{\includegraphics[width=1\textwidth]{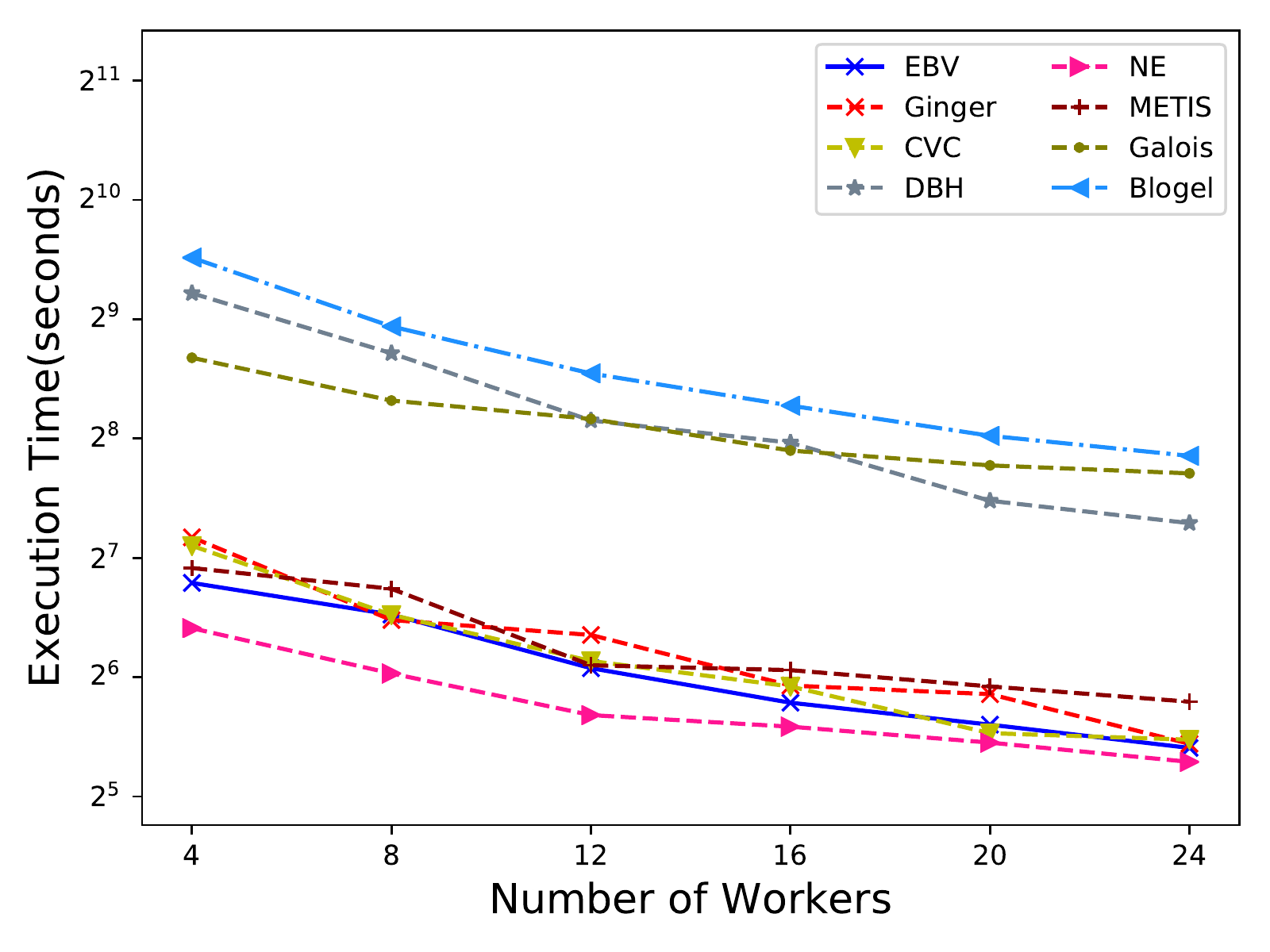}}
		\centerline{SSSP - USARoad}
	\end{minipage}
	\caption{Comparison of CC and SSSP over USARoad.\protect\\
		Figure~\ref{fig:cc_sssp_usa} shows the execution time of CC and SSSP on USARoad with different workers.
		The X-axis is the number of workers. While the Y-axis is the execution time. 
	}
	\label{fig:cc_sssp_usa}
\end{figure*}

We also compare the performance of different frameworks with CC and SSSP on the non-power-law graph (USARoad). 
Figure~\ref{fig:cc_sssp_usa} shows the performance comparison of CC and SSSP.
In this case, the performance of METIS is comparable to EBV, Ginger and CVC, which is different from Figure~\ref{fig:systemcompare}.
Moreover, NE achieves the best performance among all partition algorithms.
Based on these experiments, we can conclude that EBV outperforms other partition algorithms in the power-law graphs, and ranks closely to METIS in the non-power-law graph. 

In order to understand the mechanisms behind different partition algorithms' performance, we study the computation and communication pattern of individual subgraphs. 
Here we choose CC with $4$ workers on LiveJournal as a representative example.
Based on the subgraph-centric BSP programming model and its workflow in Section~\ref{sec:api}, a superstep can be separated into three stages: the computation stage, the communication stage and the synchronization stage. 
For the sake of convenience, we record the computation time and the communication time as $comp_i^k$ and $comm_i^k$ respectively, where $i$ indicates the subgraph ID and $k$ indicates the $k^{th}$ superstep.
The $comp$ and $comm$ are the average computation time and communication time over all workers. Thus they are computed as $\sum\limits_{i=1}^{p}\sum\limits_{k=1}^{s}comp_i^k / p $ and  $\sum\limits_{i=1}^{p}\sum\limits_{k=1}^{s}comm_i^k / p $, while $p$ and $s$ refer to the total number of workers and supersteps.
We remark that $comp + comm$ is not equal to the total execution time, but it is proportional to the later. Moreover, we define $\Delta C^k$ as $\max\limits_i(comp^k_i + comm^k_i) - \min\limits_i(comp^k_i + comm^k_i)$. $\Delta C^k$  represents the longest synchronization (waiting) time in the $k^{th}$ superstep. To demonstrate the accumulated synchronization time and its relative scale to the total execution time, we use $\Delta C = \sum\limits_{k=1}^{s} \Delta C^k$ as the metrics for the workload balance.

\begin{table}
	\centering
	\caption{Breakdown (seconds) of CC with $4$ workers over LiveJournal.}
	\label{tab:balance}
	\begin{tabular}{@{}ccccc@{}}\toprule
		& $comp$  & $comm$ & $\Delta C$ & Execution time  \\ \midrule
		EBV & 20.90 & 1.00 & 3.05 & 23.41 \\
		Ginger & 22.23 & 1.12 & 3.38 & 25.65 \\
		DBH & 23.31 & 1.31 & 7.59 & 27.85 \\ 
		CVC & 24.15 & 1.47 & 8.63 & 30.99 \\
		NE  & 17.81 & 0.52 & 28.02 & 32.66\\
		METIS & 23.28 & 0.49 & 22.70 & 34.66 \\ 
		\bottomrule
	\end{tabular}
\end{table}

\begin{figure}[htbp]
	\centering
	\begin{minipage}{0.48\linewidth}
		\centerline{\includegraphics[width=1\textwidth]{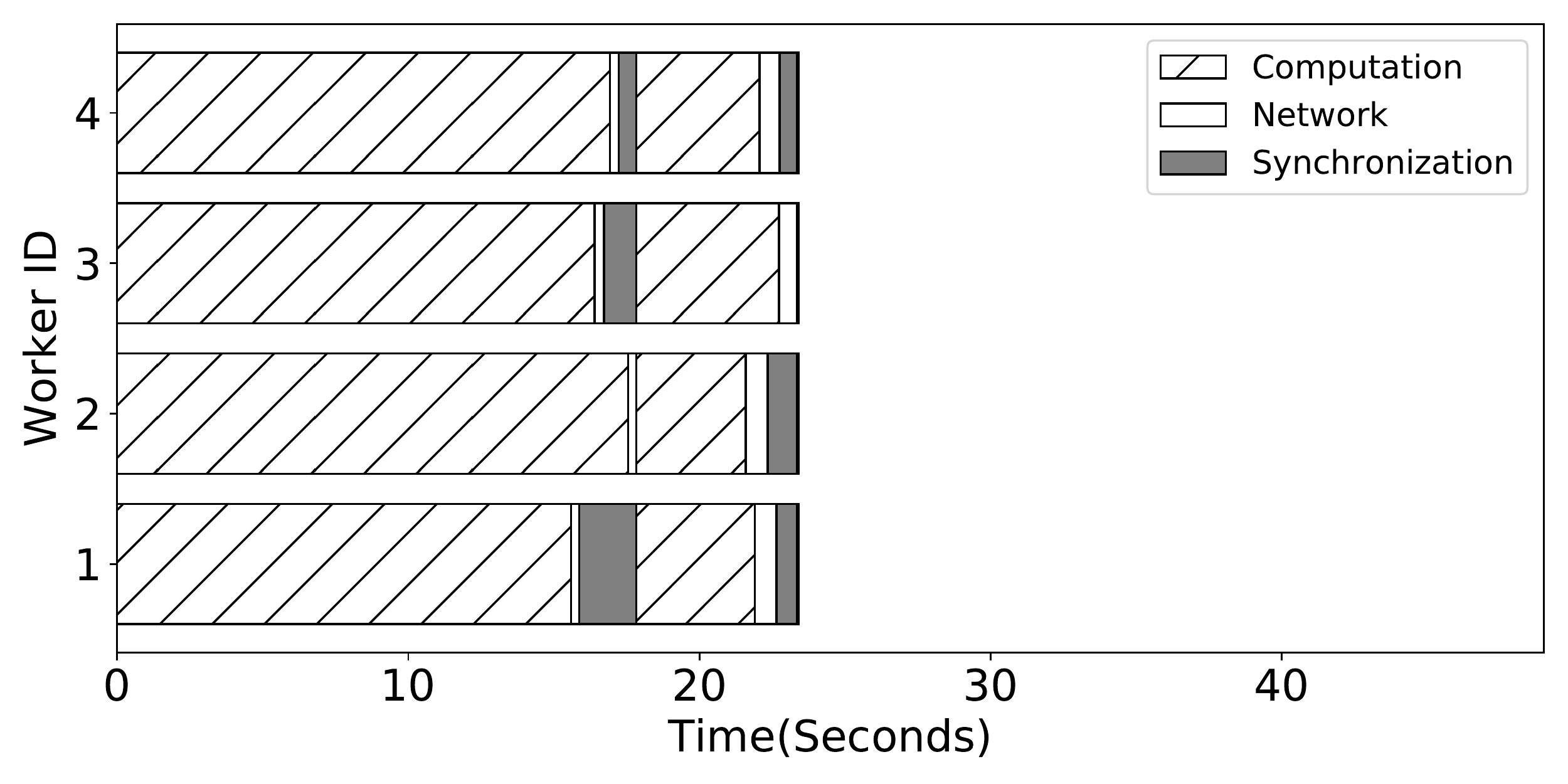}}
		\centerline{EBV}
	\end{minipage}
	\hfill
	\begin{minipage}{0.48\linewidth}
		\centerline{\includegraphics[width=1\textwidth]{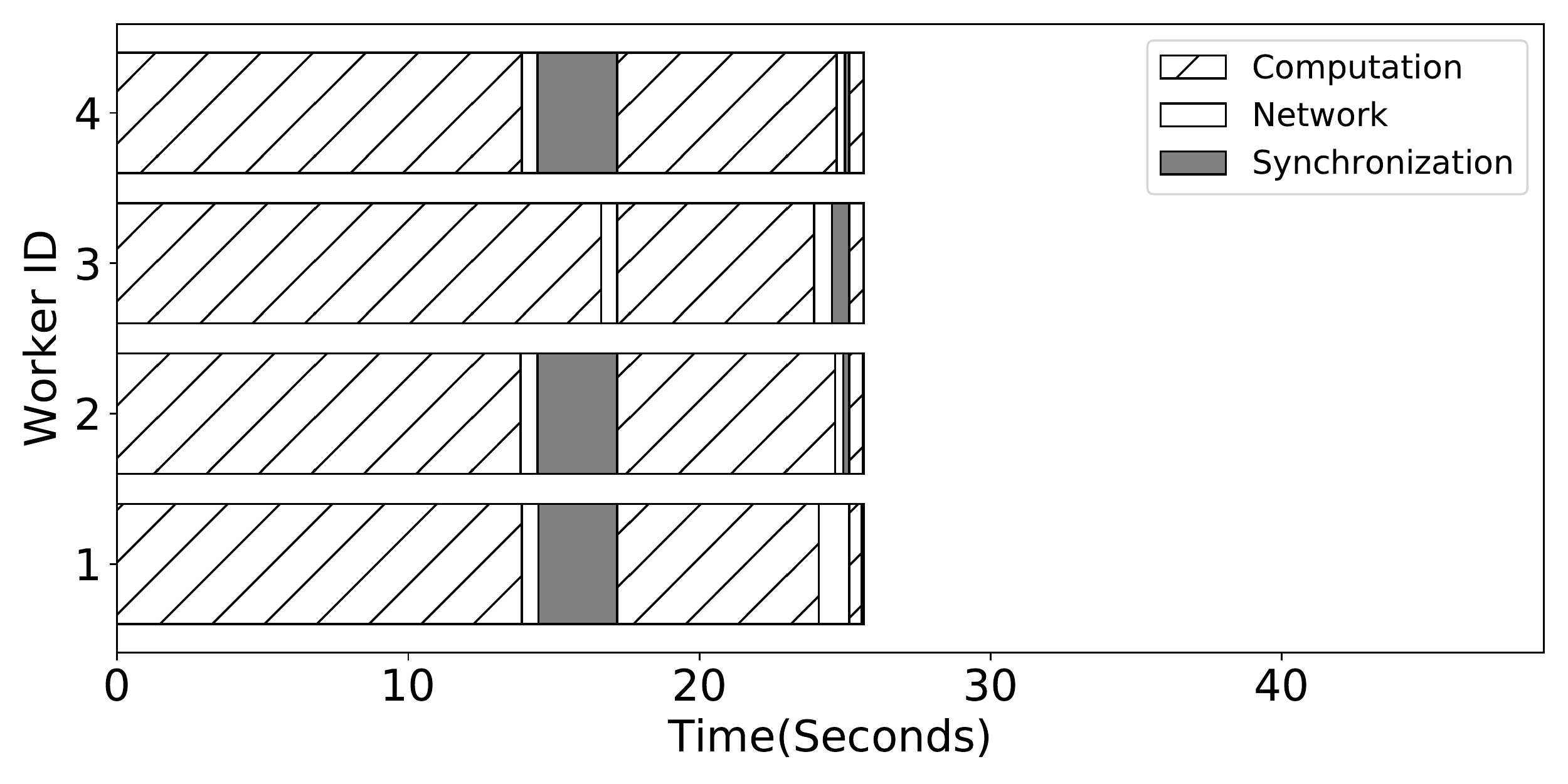}}
		\centerline{Ginger}
	\end{minipage}
	\\
	\begin{minipage}{0.48\linewidth}
		\centerline{\includegraphics[width=1\textwidth]{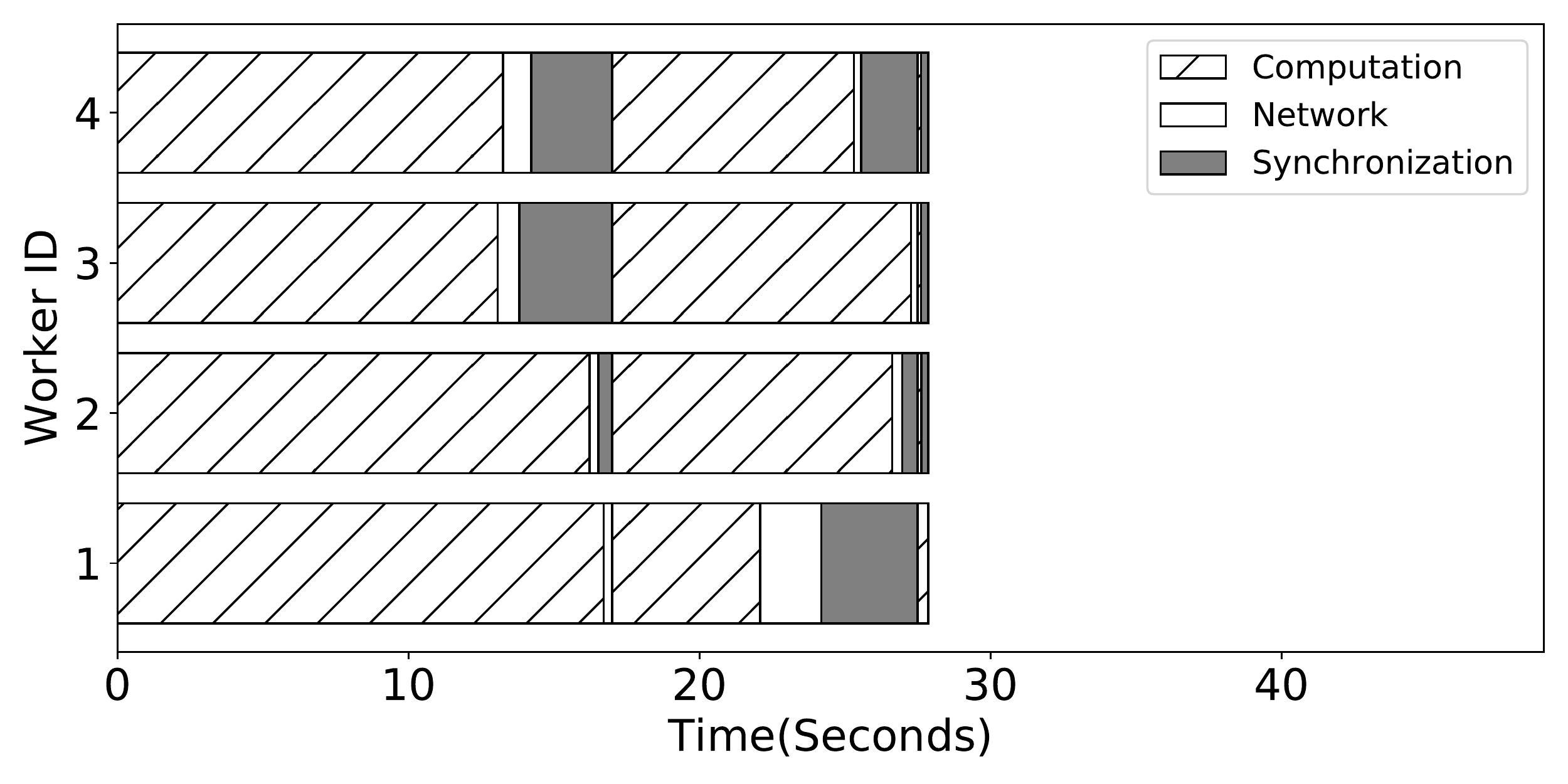}}
		\centerline{DBH}
	\end{minipage} 
    \hfill
	\begin{minipage}{0.48\linewidth}
		\centerline{\includegraphics[width=1\textwidth]{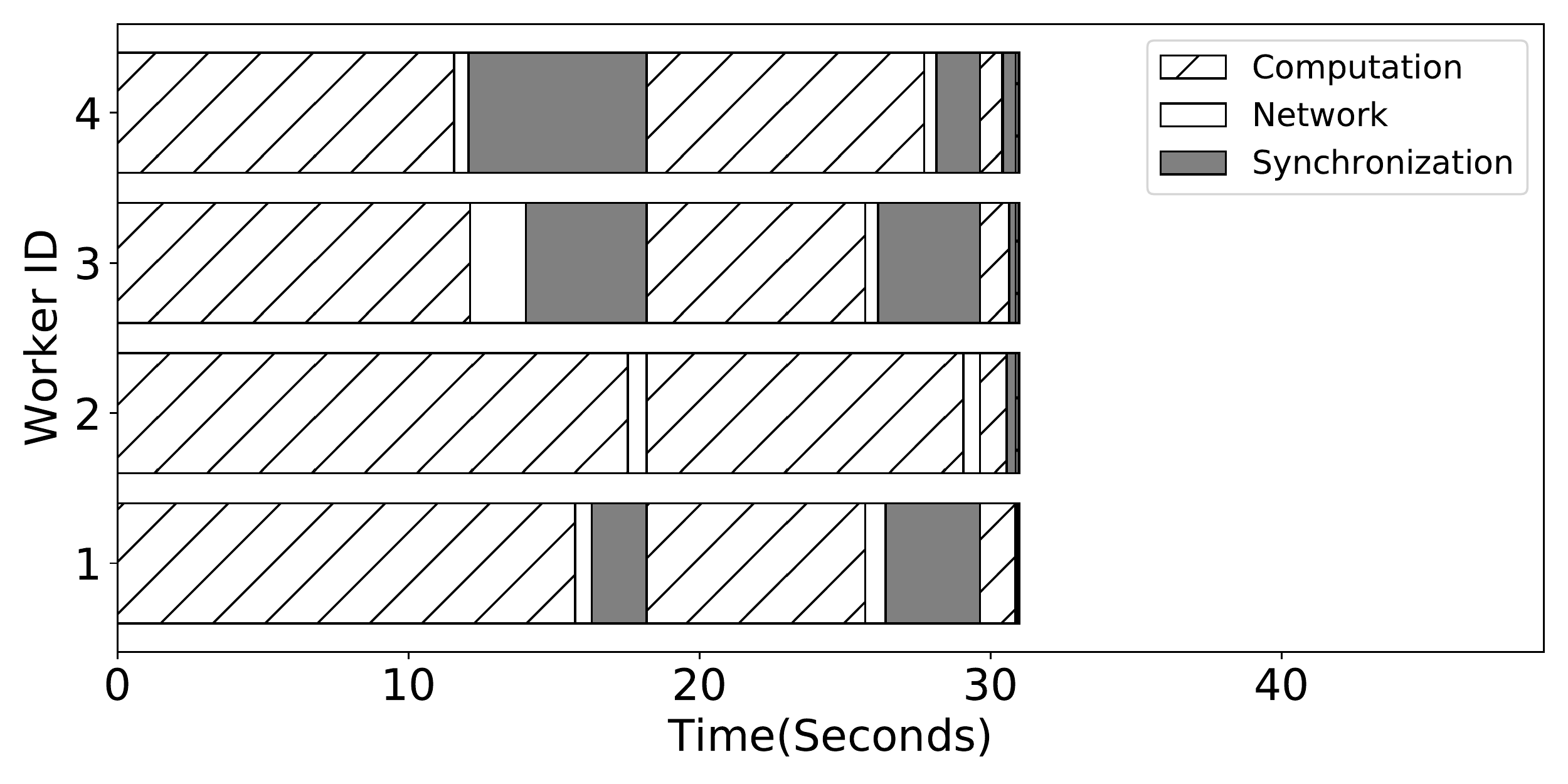}}
		\centerline{CVC}
	\end{minipage}
	\\
	\begin{minipage}{0.48\linewidth}
		\centerline{\includegraphics[width=1\textwidth]{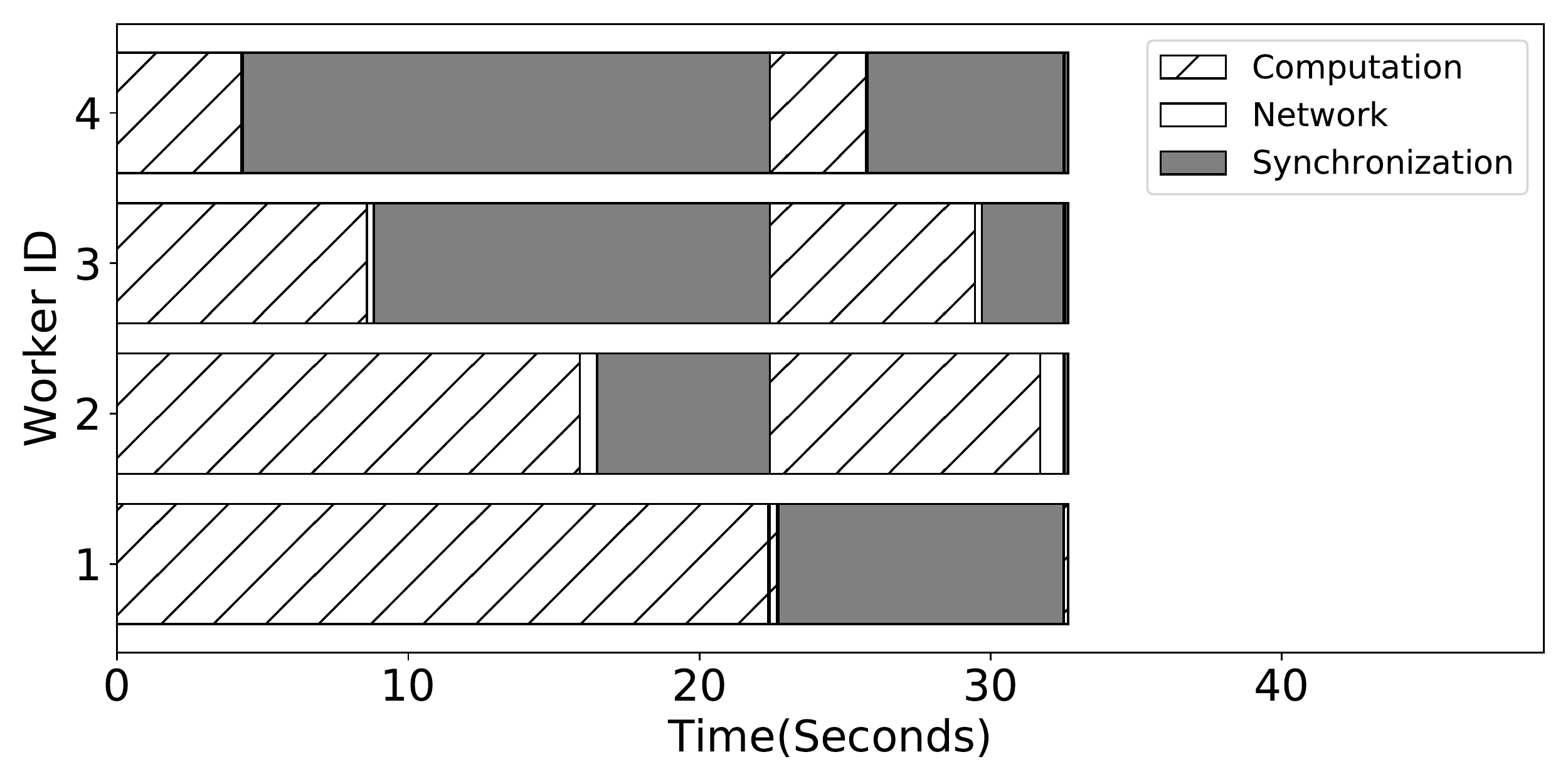}}
		\centerline{NE}
	\end{minipage} 
	\hfill
	\begin{minipage}{0.48\linewidth}
		\centerline{\includegraphics[width=1\textwidth]{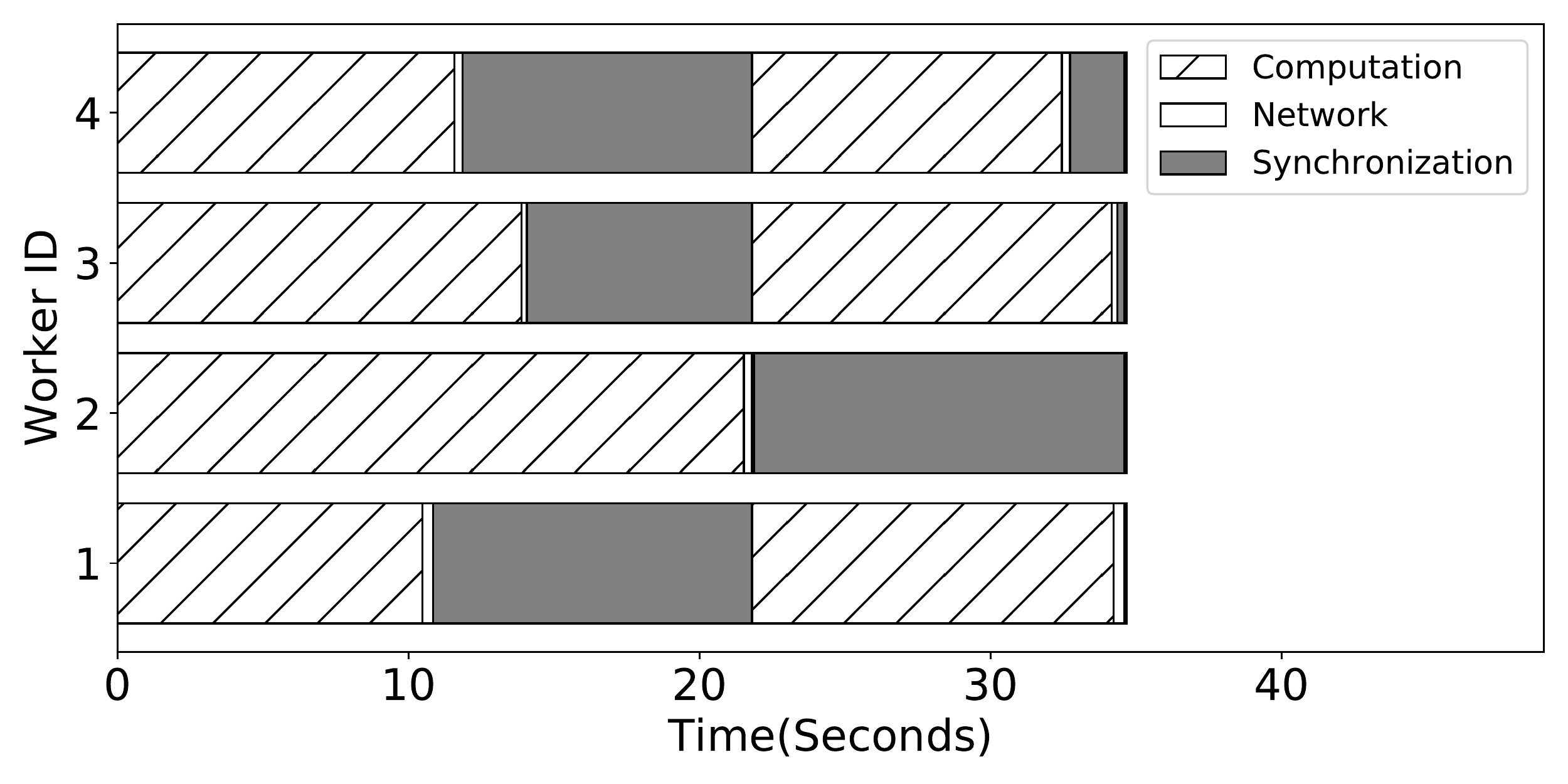}}
		\centerline{METIS}
	\end{minipage} 
	\caption{The breakdown of CC with $4$ workers over LiveJournal}
	\label{fig:balance}
\end{figure}

Table~\ref{tab:balance} demonstrates the detailed breakdown for CC with $4$ subgraphs. For the execution time, EBV is the shortest.
Although the $comp$ of METIS is similar to DBH and CVC and the $comp$ of NE is the shortest, the execution time of NE and METIS is much longer.
This is mainly because the $\Delta C$ of NE and METIS are larger than others, which indicates the strong workload imbalance.
To further demonstrate this phenomenon, we present the runtime comparison in Figure~\ref{fig:balance} for an intuitive perception.
With these experiments above, we can conclude that both the workload balance and the overall computation/communication time play extremely important roles in reducing the execution time.

\subsection{Communication Pattern Analysis}\label{sec:messages}
Here we discuss the communication characteristics of the six graph partition algorithms: EBV, Ginger, DBH, CVC, NE and METIS.
We first focus on the statistical characteristics of the partitioned subgraphs.
Further, we use the number of communication messages of each subgraph as a platform-independent indicator for quantifying the workload balance and the overall computation volume.

We present the characteristics of different graph partition algorithms with three metrics: edge imbalance factor, vertex imbalance factor and replication factor as shown in Section~\ref{sec:pre}.
The edge and vertex imbalance factors imply the computational complexity differences between subgraphs, while the replication factor is proportional to the total number of communication messages between subgraphs.

\begin{table*}[htbp]
	\centering
	\caption{Partitioning metrics comparison of EBV, Ginger, DBH, CVC, NE and METIS. \protect\\ 
		USARoad, LiveJournal, Friendster and Twitter are partitioned into $12$, $12$, $32$ and $32$ subgraphs respectively.
	}
	\label{tab:partitionercompare}
	\centering
	\begin{tabular}{cC{15 pt} C{30 pt}C{30 pt}C{30 pt}C{30 pt}C{30 pt}C{30 pt}  C{17 pt}C{19 pt}C{17 pt}C{17 pt}C{17 pt}C{24 pt}}\toprule
		& & \multicolumn{6}{c}{Edge Imbalance Factor / Vertex Imbalance Factor} & \multicolumn{6}{c}{Replication Factor} \\
		\cmidrule(lr){3-8}  \cmidrule(lr){9-14} 
		Graphs           & $\eta$ &  EBV & Ginger & DBH  & CVC & NE  & METIS  & EBV & Ginger & DBH & CVC & NE & METIS \\ \midrule
		USARoad   & 6.30 & 1.00/1.00 & 1.00/1.00 & 1.00/1.00 & 1.00/1.00 & 1.00/1.05 & 1.10/1.01 & 1.29 & 1.68 & 1.92 & 2.27 & 1.01 & 1.00  \\
    	LiveJournal & 2.64 & 1.00/1.01 & 1.00/1.01 & 1.00/1.00 & 1.00/1.00 & 1.00/2.14 & 2.07/1.03 & 1.80 &2.23 & 3.34 & 3.96 & 1.89 & 1.20  \\
    	Friendster  & 2.43 &1.00/1.00 & 1.04/1.10 & 1.00/1.00 & 1.00/1.00 & 1.00/2.46 & 2.43/1.03 & 4.63 & 5.64 &6.55 & 5.95 & 4.43 & 1.36  \\
		Twitter   & 1.87 &1.00/1.01 & 1.02/1.03 & 1.00/1.00 & 1.03/1.04 & 1.00/3.64 & 6.44/1.05 & 3.59 & 4.51 & 4.99 & 6.75 & 2.42 & 1.56  \\
		\bottomrule
	\end{tabular}
\end{table*}

Table~\ref{tab:partitionercompare} shows the comparison of the edge and vertex imbalance factors and the replication factor among all algorithms.
For the convenience of analysis, we also provide the $\eta$ value and order this table by $\eta$ reversely.
Combing with the $\eta$ value, we find that as $\eta$ decreases (more skewed), the partition results of NE and METIS are more imbalanced.
For the non-power-law graph (with the largest $\eta$), NE and METIS achieve the roughly balanced results.
This phenomenon explains why they perform better in the non-power-law graph as shown in Figure~\ref{fig:cc_sssp_usa}.
Benefit from the evaluation function mentioned in Section~\ref{sec:greedy}, EBV outperforms Ginger, DBH and CVC for the replication factor.
Both the edge and vertex imbalance factors of EBV, Ginger, DBH and CVC are almost $1$.
The good performance of EBV for the edge and vertex imbalance factors and the replication factor conform to its performance in Figure~\ref{fig:systemcompare}. 

We further discuss the metrics shown above by combing with the number of communication messages.
For the sake of simplicity, we focus on the experimental data of CC here.
The number of workers for USARoad, LiveJournal, Friendster and Twitter are $12$, $12$, $32$, $32$, which is the same as Table~\ref{tab:partitionercompare}.
Table~\ref{tab:message_sum} shows the total number of communication messages during the computation for CC.
According to this data, NE and METIS produce a small number of messages and lead by a large margin on the non-power-law graph (USARoad).
The total number of EBV's communication messages is smaller than Ginger, DBH and CVC all the time.
This is in accordance with the replication factor of EBV, Ginger, DBH and CVC in Table~\ref{tab:partitionercompare}. 
More specifically, EBV can reduce $23.7\%$, $23.8\%$, $35.4\%$ and $26.0\%$ messages over USARoad, LiveJournal, Friendster and Twitter than Ginger respectively.
Overall, the data for the total number of messages in the parallel graph computations confirms the tendency of the replication factor.

\begin{table*}[htbp]
	\centering
	\caption{
		The total number of communication messages on the CC algorithm. \protect\\
		This table shows the total number of communication messages between all workers for EBV, Ginger, DBH, CVC, NE and METIS on four graphs.
		The numbers in parentheses are the corresponding replication factors in Table~\ref{tab:partitionercompare}.
	}
	\label{tab:message_sum}
	\begin{tabular}{@{}ccccccc@{}}\toprule
		Graphs  & EBV & Ginger &DBH & CVC & NE & METIS \\ \midrule
		USARoad  & $4.05 \times 10 ^7$ (1.29) & $5.01 \times 10 ^7$ (1.68) & $5.41 \times 10 ^7$ (1.92) & $1.26 \times 10 ^8$ (2.27) & $3.14 \times 10 ^5$ (1.01) & $1.63 \times 10 ^4$ (1.00) \\
		LiveJournal  & $1.43 \times 10 ^7$ (1.80) & $1.77 \times 10 ^7$ (2.23)  & $1.85 \times 10 ^7$ (3.34) & $2.18 \times 10 ^7$ (3.96) & $8.36 \times 10 ^6$ (1.89) & $8.97 \times 10 ^6$ (1.20) \\ 
		Friendster  & $3.45 \times 10 ^8$ (4.63) & $4.67 \times 10 ^8$ (5.64)& $6.97 \times 10 ^8$ (6.55) & $4.51 \times 10 ^8$ (5.95) & $3.31 \times 10 ^8$ (4.43) & $2.64 \times 10 ^8$ (1.36) \\ 
		Twitter   & $1.81 \times 10 ^8$ (3.59) & $ 2.28 \times 10 ^8$ (4.51) & $2.39 \times 10 ^8$ (4.99) & $3.11 \times 10 ^8$ (6.75) & $9.98 \times 10 ^7$ (2.42) & $1.52 \times 10 ^8$ (1.56) \\ 
		\bottomrule
	\end{tabular}
\end{table*}

Although NE and METIS provide better performance in the total communication volume than EBV, it can be observed that EBV outperforms NE and METIS on the power-law graphs. 
The better performance of EBV against Ginger, DBH and CVC can be attributed to the reduction of the communication. 
But we can not make the same claim for NE and METIS. 
Thus, we further analyze the characteristics of the workload balance with communication messages.
For quantifying the imbalance of communication messages for all workers, we introduce the $max/mean$ metrics. 
The $max/mean$ metrics means the ratio of the maximum number of messages sent by each worker and the average number of messages.
Since the overall execution time is denoted by the slowest worker, this metrics is more meaningful than the common variance metrics.
Table~\ref{tab:message_var} presents the $max/mean$ ratio of the number of messages on the CC algorithm.
The corresponding edge and vertex imbalance factors are also presented in parentheses.
Coupling with the $max/mean$ metrics and the imbalance factor, we can see that the $max/mean$ value of EBV, Ginger, DBH and CVC are almost $1$ on all graphs just as the imbalance factors.
Meanwhile, the $max/mean$ values of NE and METIS vary a lot and is extremely affected by the vertex or edge imbalance factors.
In most cases, the $max/mean$ value increases as the edge or vertex imbalance factor gets larger.
It reveals the inefficiency reason for NE and METIS on power-law graphs.

\begin{table*}[htb]
	\centering
	\caption{The $max/mean$ ratio of the number of messages on the CC algorithm \protect\\
		This table shows the $max/mean$ ratio of the number of messages among all workers for EBV, Ginger, DBH, CVC, NE and METIS.
		We also provide the corresponding edge and vertex imbalance factors (Table~\ref{tab:partitionercompare}) in parentheses.
	}
	\label{tab:message_var}
	
	\begin{tabular}{@{}ccccccc@{}}\toprule
		Graphs & EBV & Ginger & DBH & CVC & NE & METIS \\ \midrule
		USARoad & 1.002 (1.00/1.00) & 1.001 (1.00/1.00) & 1.001 (1.00/1.00) &1.019 (1.00/1.00) & 2.670 (1.00/1.05) & 1.75 (1.10/1.01)\\
		LiveJournal & 1.002 (1.00/1.01) & 1.005 (1.00/1.01) & 1.002 (1.00/1.00) &1.008 (1.00/1.00) & 1.697 (1.00/2.14) & 1.93 (2.07/1.03)\\ 
		Friendster & 1.000 (1.00/1.00) & 1.005 (1.04/1.10) & 1.001 (1.00/1.00)  &1.018 (1.00/1.00) & 1.623 (1.00/2.46) & 2.12 (2.43/1.03)\\ 
		Twitter & 1.001 (1.00/1.01) & 1.013 (1.02/1.03) & 1.001 (1.00/1.00)  &1.067 (1.03/1.04) & 2.451 (1.00/3.64) & 3.29 (6.44/1.05)\\ 
		\bottomrule
	\end{tabular}
\end{table*} 

\subsection{Sorting Analysis}\label{sec:sorting}
Finally, we want to investigate the effort of the edge sorting preprocessing. 
We mark our EBV algorithm with or without the sorting preprocessing as EBV-sort and EBV-unsort respectively.
We test EBV-sort and EBV-unsort over three power-law graphs (LiveJournal, Twitter and Friendster) while partitioning them into $4$, $8$, $16$ and $32$ subgraphs.
From Figure~\ref{fig:sort-compare}, we find that on all three power-law graphs, EBV-sort outperforms EBV-unsort for the final replication factor.
Besides, as the number of subgraphs grows, the margin between EBV-sort and EBV-unsort gets larger.
It shows that the sorting preprocessing improves the performance of EBV, especially when the number of subgraphs is large.

Further, we want to study the shapes of these curves.
The replication factor curves of EBV-sort increase sharply at the beginning, and tend to a fixed value later.
According to the definition of replication factor in Section~\ref{sec:pre}, it is proportional to the total number of vertices in all subgraphs.
Therefore, we can conclude that the EBV-sort algorithm assigns edges with low end-vertices' degrees to subgraphs and create many vertices at the beginning.
When processing the edges with high end-vertices' degrees, almost no new vertices are created.
This property implies that the EBV-sort algorithm has great potential in many dense power-law graphs.

\begin{figure*}
	\centering
	\begin{minipage}{0.32\linewidth}
		\centerline{\includegraphics[width=1\textwidth]{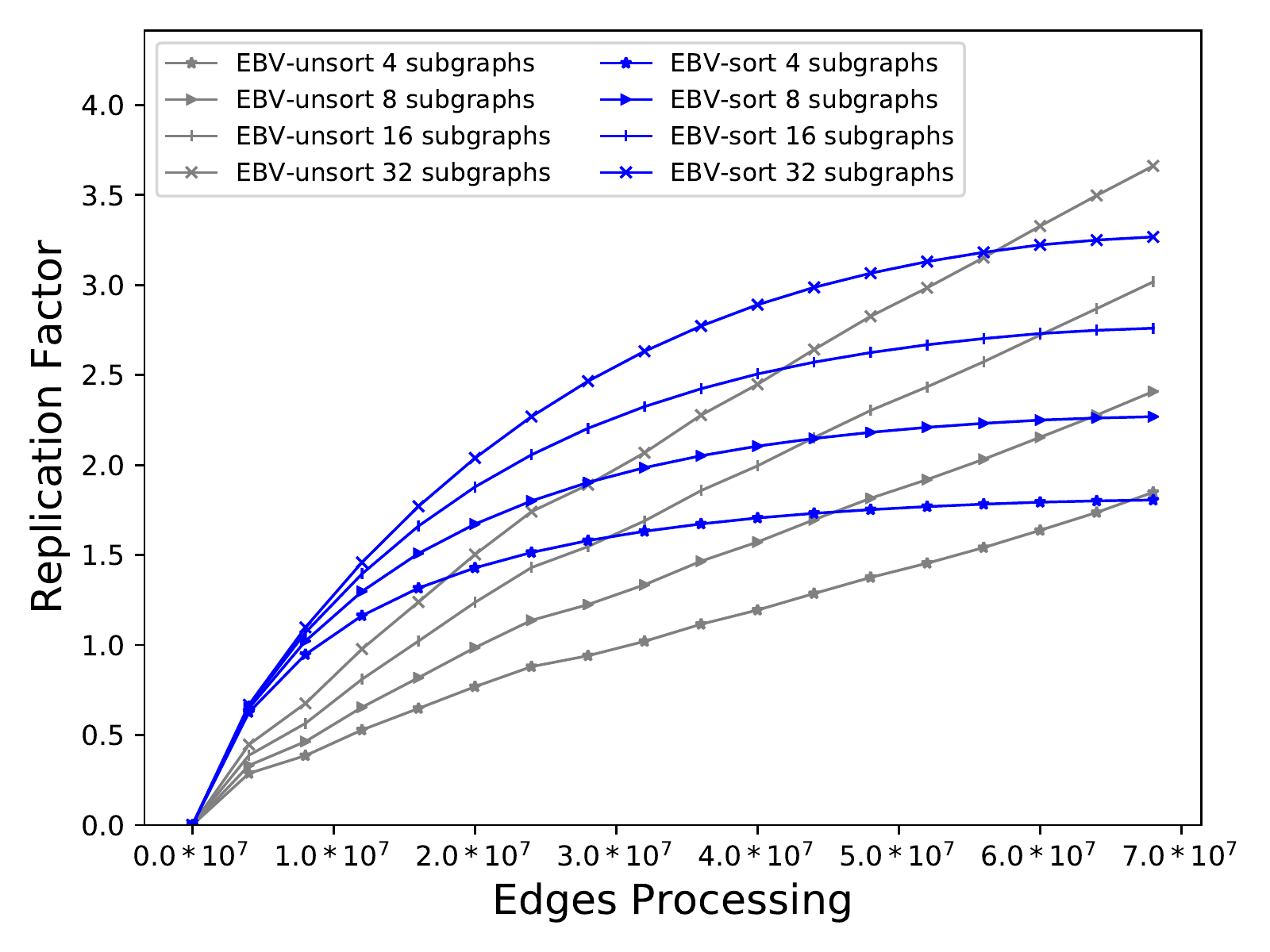}}
		\centerline{LiveJournal}
	\end{minipage}
	\hfill
	\begin{minipage}{0.32\linewidth}
		\centerline{\includegraphics[width=1\textwidth]{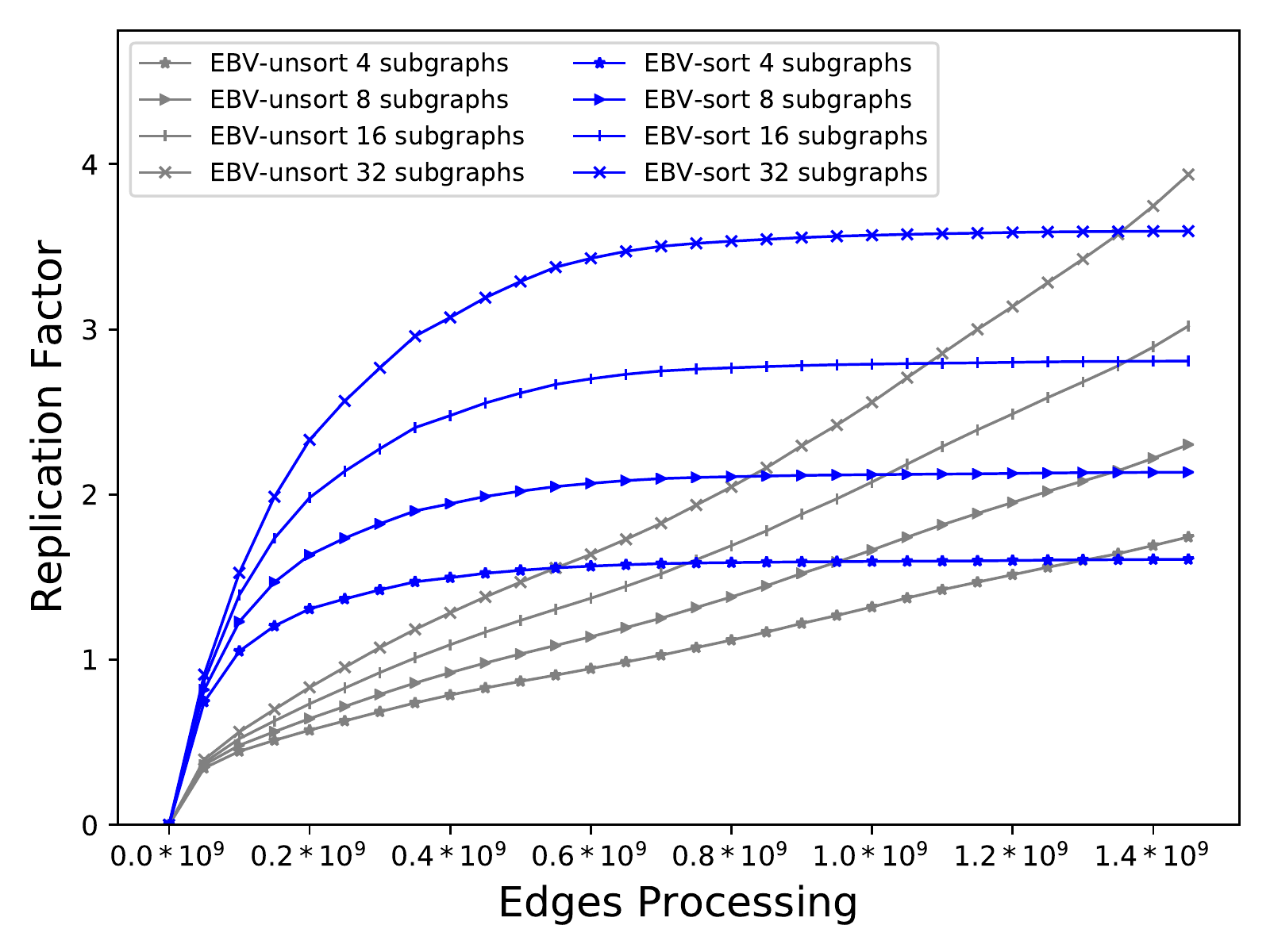}}
		\centerline{Twitter}
	\end{minipage}
	\hfill
	\begin{minipage}{0.32\linewidth}
		\centerline{\includegraphics[width=1\textwidth]{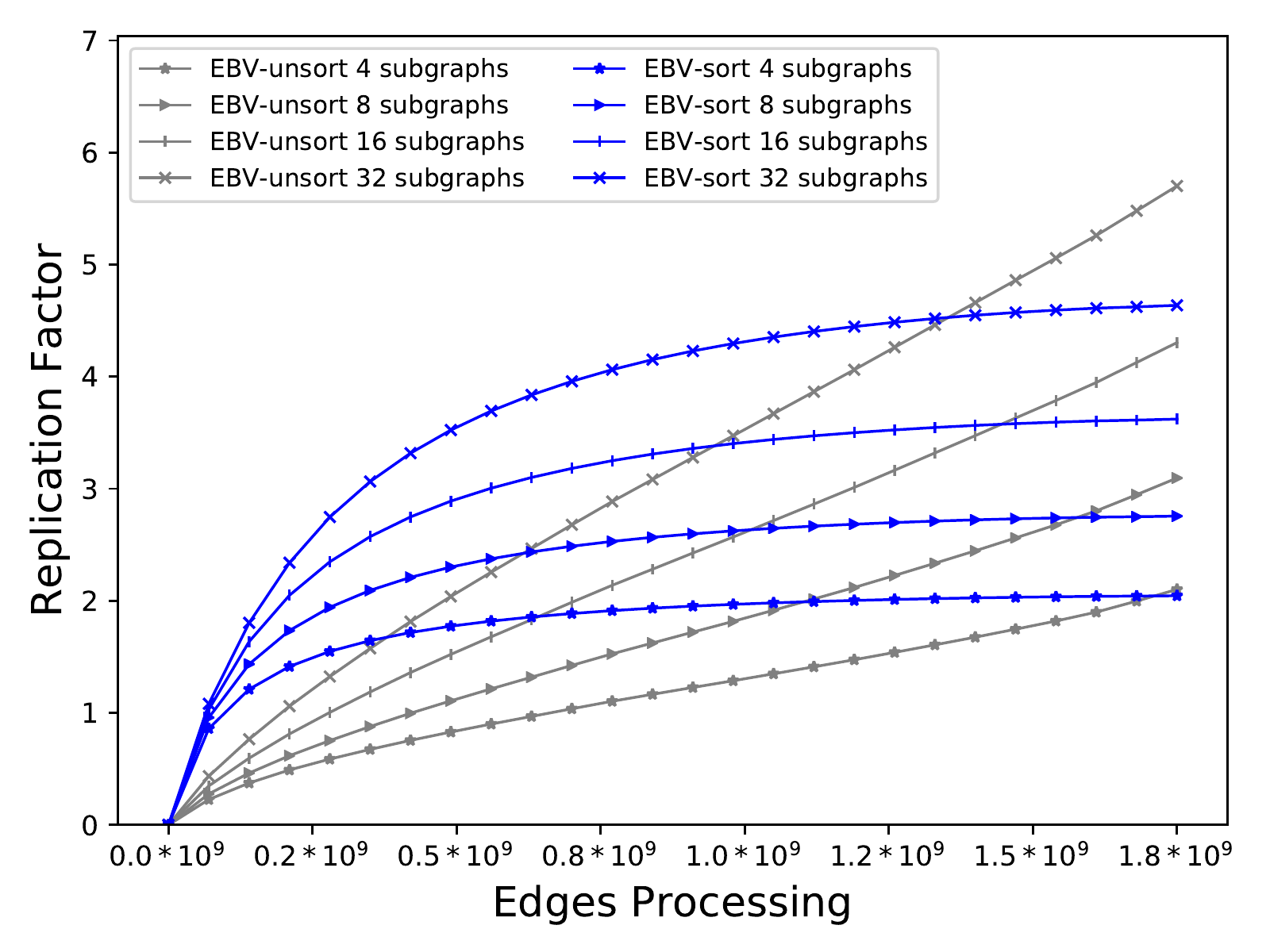}}
		\centerline{Friendster}
	\end{minipage}

	\caption{Replication Factor Growth Curve.\protect\\
		Figure~\ref{fig:sort-compare} shows the growth curve of replication factor over LiveJournal, Twitter and Friendster with $4$, $8$, $16$ and $32$ subgraphs. 
		The X-axis is the number of edges have been assigned and the Y-axis is the replication factor of current results.
	}
	\label{fig:sort-compare}
\end{figure*}

\subsection{Summary}
With the experiments, we can answer these questions we proposed:
\begin{enumerate}[(1)]
	\item We select five representative partition algorithms: Ginger, DBH, CVC, NE and METIS to compare with EBV. 
	For NE and METIS, they only minimize the total message size. For Ginger, DBH and CVC, they care more about the message imbalance and partly ignore the total message size. For EBV, it optimizes both factors.
	Experiments in Figure~\ref{fig:systemcompare} show that EBV outperforms other partition algorithms, which indicates that both the total message size and the message imbalance are essential for partitioning power-law graphs.
	\item Table~\ref{tab:message_sum} and Table~\ref{tab:message_var} reveal that the edge and vertex imbalance factors and replication factor are closely related to the message balance and the total size of messages. The experiment in Figure~\ref{fig:systemcompare} further proves that EBV performs better than Ginger, DBH, CVC, NE, METIS and other representative frameworks in large-scale power-law graphs, which is attributed to the well-designed evaluation function that considers these two factors.
	\item We compare the replication factor of EBV-sort and EBV-unsort in power-law graphs with different numbers of subgraphs. Experiments in Figure~\ref{fig:sort-compare} show that the sorting preprocessing can reduce the replication factor of EBV in power-law graphs significantly.
\end{enumerate}

\subsection{Limitation}
Many subgraph-centric frameworks do not publish their source code such as Giraph++~\cite{tian2013think} and GRAPE~\cite{fan2017parallelizing}. Thus they are not included in our cross-platform comparison. However, it does not influence our conclusion. Since the partition methods they employed are presented in our experiments.

%% file: relatedwork.tex
\section{Related Work}\label{sec:relatedwork}
Here we review the graph partition algorithms that can handle large-scale graphs. Some spectral algorithms that utilize the global graph information such as \cite{alpert1999spectral,mcsherry2001spectral} are omitted, because of their heavy partition overhead.

\underline{Local-Based:} 
The most famous local-based graph partition algorithm is METIS~\cite{karypis1997parmetis}. METIS applies several stages to coarsen and uncoarsen the vertices and edges. It also use the K-L algorithm~\cite{kernighan1970efficient} for refining its results during the uncoarsening phase.
Andersen et al.~\cite{andersen2006local} present a local-based partitioning algorithm using a variation of pagerank with a specified starting distribution.
They show that the ordering of the vertices produced by a pagerank vector reveals a cut with small conductance.
A recent local-based graph partition algorithm is NE~\cite{ne}.
It expands the core vertex from the boundary vertex set round by round to partition graphs.
Moreover, they propose a distributed version of NE for trillion-edge graphs~\cite{NEdistributed}.
Shengwei Ji proposes a two-stage heuristic method~\cite{local}, which selects high-degree vertices as core vertices and than expands other vertices.
Those local-based algorithms can produce partition results with small replication factor. 
However, they are limited for the balance of both edges and vertices on power-law graphs.


\underline{Self-Based:} 
The major approach for self-based graph partition algorithms is based on random hash.
For the edge-cut method, it just assigns the vertex by hashing the vertex's ID.
For the vertex-cut method, one straight approach is hashing the edge with its end-vertices' ID into a 1-dimensional value.
Another approach is splitting the adjacency matrix into several blocks (2D partitioning), such as CVC~\cite{boman2013scalable}.
For better handling the skewed degree distribution of power-law graphs, the degree of vertices is widely used.
DBH~\cite{xie2014distributed} cut vertices with higher degrees to obtain better performance.
The hybrid-cut~\cite{chen2019powerlyra} is also proposed with similar ideas.
They assign the edges with the same low-degree target vertex to the same subgraph.
For edges with high-degree target vertices, they assign them according to the source vertices.
Those algorithms are lightweight and can produce roughly balanced results naturally.
However, their partition results need to be further improved for reducing the replication factor.
Besides, they ignore the influence of the processing order for greedy algorithms.

Recently, streaming graph partition algorithms are popular for partitioning large-scale graphs and most of them are self-based algorithms, such as Fennel~\cite{Fennel}, HDRF~\cite{HDRF} and CuSP~\cite{hoang2019cusp}.
They view the input graph as a sequence of edges and process them with just one-pass without extra information.
ADWISE~\cite{mayer2018adwise} is a compromise between streaming and offline, they improve the partitioning quality by caching the current processing edges with the adaptive window size.
However, their performance is limited due to the lack of information on the entire graph.






%% file: conclusion.tex
\section{Conclusion and Future Work}\label{sec:conclusion}

In this paper, we present an efficient and balanced vertex-cut graph partition algorithm for the subgraph-centric model. It is devised to improve the performance of frameworks with subgraph-centric model for large-scale power-law graphs. 
Our results show that graphs partitioned by EBV have the same level of the communication imbalance compared with Ginger, DBH and CVC, while the total communication volume is on the same scale as NE and METIS for power-law graphs. As a result, the EBV algorithm outperforms Ginger, DBH, CVC, NE and METIS.
Experiments indicate that the EBV algorithm achieves outstanding performance for large-scale power-law graphs.
Based on these experiments and analysis, we conclude that the EBV algorithm has very good potential for processing large-scale power-law graphs far beyond the test examples we presented in this paper.

Our algorithm can be further improved in several aspects.
Firstly, EBV is a sequential and offline partition algorithm. We might need to extend it to the distributed and streaming environment to handle larger graphs.
Secondly, a more complex sorting mechanism can be tried for further improving its performance.

In the future, we plan to apply EBV to distributed graph neural networks (GNN) for processing large graphs.
We also plan to explore other potential optimization strategies to form a partition algorithm which could reduce the total communication volume and the communication imbalance further, while it maintains a reasonable partition overhead.